\documentclass[conference]{IEEEtran}
\IEEEoverridecommandlockouts

\usepackage{cite}
\usepackage{url}
\usepackage{amssymb}           

\usepackage{amsthm}
\usepackage{amsmath,amssymb,amsfonts}

\usepackage[linesnumbered,ruled,vlined]{algorithm2e}
\usepackage{graphicx}
\usepackage{textcomp}
\usepackage{balance}
\usepackage{xcolor}
\usepackage{subfig}
\usepackage{multirow}
\usepackage{makecell}
\usepackage{siunitx}
\usepackage{color}
\usepackage{colortbl}
\def\BibTeX{{\rm B\kern-.05em{\sc i\kern-.025em b}\kern-.08em
    T\kern-.1667em\lower.7ex\hbox{E}\kern-.125emX}}
\usepackage{hyperref}

\newcommand{\deepcolor}{\cellcolor[rgb]{0.8745, 0.8235, 0.7137}}
\newcommand{\lightcolor}{ \cellcolor[rgb]{0.9882,0.9373,0.8196}}

\newtheorem{theorem}{\textbf{Theorem}}

\SetKwInOut{Input}{Input}
\SetKwInOut{Output}{Output}
\newcommand{\eat}[1]{}
\newcommand{\stitle}[1]{\noindent{\bf #1}}
\title{Fast Tuning the Index Construction Parameters of Proximity Graphs in Vector Databases}

\author{Wenyang~Zhou,
        Jiadong~Xie,
        Yingfan~Liu,
        Zhihao~Yin,
        Jeffrey~Xu~Yu,\\
        Hui~Li,
        Zhangqian~Mu,
        Xiaotian~Qiao,
        and Jiangtao~Cui

\thanks{This work was supported by projects funded by National Natural Science Foundation of China (NSFC) under Grants 62002274. (Corresponding authors: Yingfan Liu)}

\thanks{Wenyang Zhou, Yingfan Liu, Zhihao Yin, Hui Li, Zhangqian Mu, Xiaotian Qiao,and Jiangtao Cui are with the School of Computer Science and Technology, Xidian University, Xi'an 710126, China (e-mail: wenyangchou@outlook.com; liuyingfan@xidian.edu.cn)}

\thanks{Jiadong Xie is with the Chinese University of Hong Kong, Hong Kong, China (e-mail: jdxie@se.cuhk.edu.hk).}

\thanks{Jeffrey Xu Yu is with the Hong Kong University of Science and Technology (Guangzhou), Guangzhou, China (e-mail: jeffreyxuyu@hkust-gz.edu.cn)}
}

\begin{document}

\maketitle

\begin{abstract}
$k$-approximate nearest neighbor search ($k$-ANNS) in high-dimensional vector spaces is a fundamental problem across many fields. With the advent of vector databases and retrieval-augmented generation, $k$-ANNS has garnered increasing attention. Among existing methods, proximity graphs (PG) based approaches are the state-of-the-art (SOTA) methods. However, the construction parameters of PGs significantly impact their search performance. Before constructing a PG for a given dataset, it is essential to tune these parameters, which first recommends a set of promising parameters and then estimates the quality of each parameter by building the corresponding PG and then testing its $k$-ANNS performance. Given that the construction complexity of PGs is superlinear, building and evaluating graph indexes accounts for the primary cost of parameter tuning. Unfortunately, there is currently no method considered and optimized this process.
In this paper, we introduce FastPGT, an efficient framework for tuning the PG construction parameters. FastPGT accelerates parameter estimation by building multiple PGs simultaneously, thereby reducing repeated computations. Moreover, we modify the SOTA tuning model to recommend multiple parameters at once, which can be efficiently estimated using our method of building multiple PGs simultaneously. 
Through extensive experiments on real-world datasets, we demonstrate that FastPGT achieves up to 2.37x speedup over the SOTA method VDTuner, without compromising tuning quality.
\end{abstract}


\section{Introduction}
With the breakthrough in learning-based embedding models, objects such as text chunks~\cite{word2vec} and images~\cite{nasrabadi1988image} can be embedded into high-dimensional vectors, capturing their semantics. This advancement has shifted the task of finding similar objects to $k$-approximate nearest neighbor ($k$-ANN) search ($k$-ANNS) in high-dimensional spaces, which has been a fundamental problem in fields such as information retrieval~\cite{HuangSSXZPPOY20,li2021embedding}, recommendation systems~\cite{OkuraTOT17}, and retrieval-augmented generation~\cite{REALM,RAG-ACL,liu2024retrievalattentionacceleratinglongcontextllm,deng2025alayadb,zhang2025pqcache}.
Specifically, given a dataset $D$ containing vectors in $\mathbb{R}^d$ space and a query vector \( q \in \mathbb{R}^d \), $k$-ANNS aim at returning $k$ vectors that are sufficiently close to the query vector $q$ among vectors in $D$.

During the past decades, a bulk of $k$-ANNS approaches have been proposed in the literature, including tree based indexes~\cite{Kdtree,Rtree,Xtree,Srtree}, hash based indexes~\cite{Datar2004,MPlsh,LSB,C2LSH,SKLSH}, invert index based indexes~\cite{IMI,PQ}, and proximity graph (PG) based indexes (PG)~\cite{Hnsw,Nsg,diskann,FastPG,ALMG,liu2025privacy}. 
According to recent works~\cite{Dpg, survey2021, PGsurvey25}, PG based methods significantly outperform other methods in $k$-ANNS performance. 
Therefore, in this paper, we focus on the PG-based approaches for $k$-ANNS.
%
Existing PG-based methods could be divided into three categories~\cite{FastPG}: (1) $k$-nearest neighbor graph (KNNG)~\cite{KGraph,MyKNNGCsurvey} that simply builds directed edges from each vector in $D$ to its $k$-ANN; (2) relative neighborhood graph (RNG)~\cite{Nsg, Nssg, diskann, taumg, ALMG} assumes the full knowledge of $D$ and builds undirected edges between each vector and its pruned set of close neighbors via various pruning strategies; and (3) navigable small world graph (NSWG)~\cite{Hnsw,Sw} that assumes no prior knowledge of $D$ and builds the graph by inserting each vector into the graph one by one. In general, KNNG has the worst $k$-ANNS performance compared to RNG and NSWG~\cite{Dpg, survey2021}. Hence, in this paper, we primarily study the issues with NSWG and RNG.
%
%
For all PG-based methods, their $k$-ANNS performance is notably sensitive to their index construction parameters. Take, for instance, HNSW, a representative NSWG method shown in Figure~\ref{fig:cp_sensitivity-a}, which has two construction parameters: $M$ for the out-degree limit in its graph index and $efc$ for the search pool size. Similarly, Vamana, a representative RNG method depicted in Figure~\ref{fig:cp_sensitivity-b}, involves three construction parameters: $M$, serving as the out-degree limit in its graph index, $L$ for the search pool size, and $\alpha$ for its pruning strategy. The results show both are significantly affected by their construction parameters, underscoring the importance of tuning the index construction parameters of PGs for optimal search performance in practice.

\begin{figure}[t]
\centering
\subfloat[HNSW on Sift]{\includegraphics[width=0.47\linewidth]{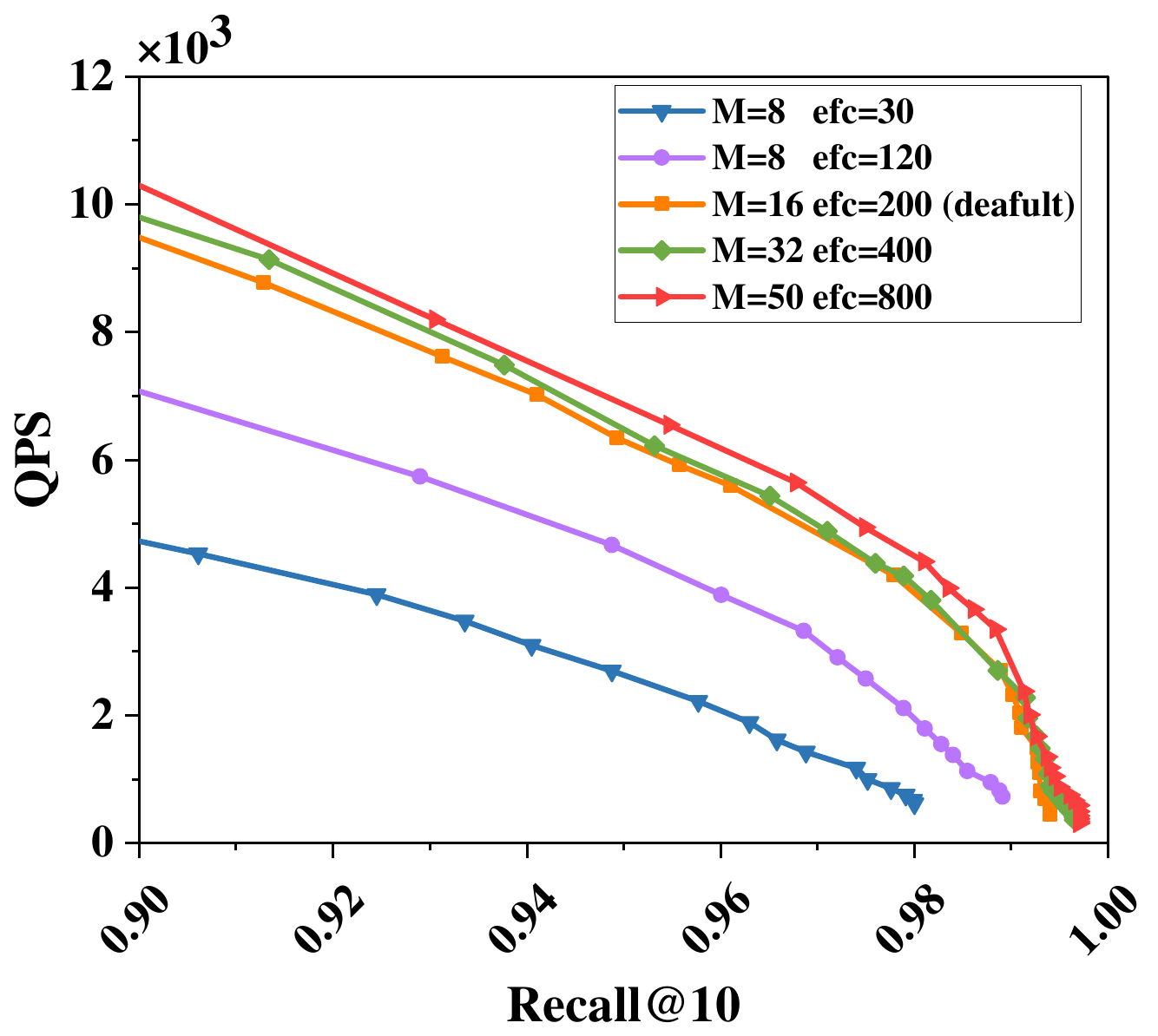}
\label{fig:cp_sensitivity-a}}
\hspace{0.1cm}
\subfloat[Vamana on Sift]{\includegraphics[width=0.47\linewidth]{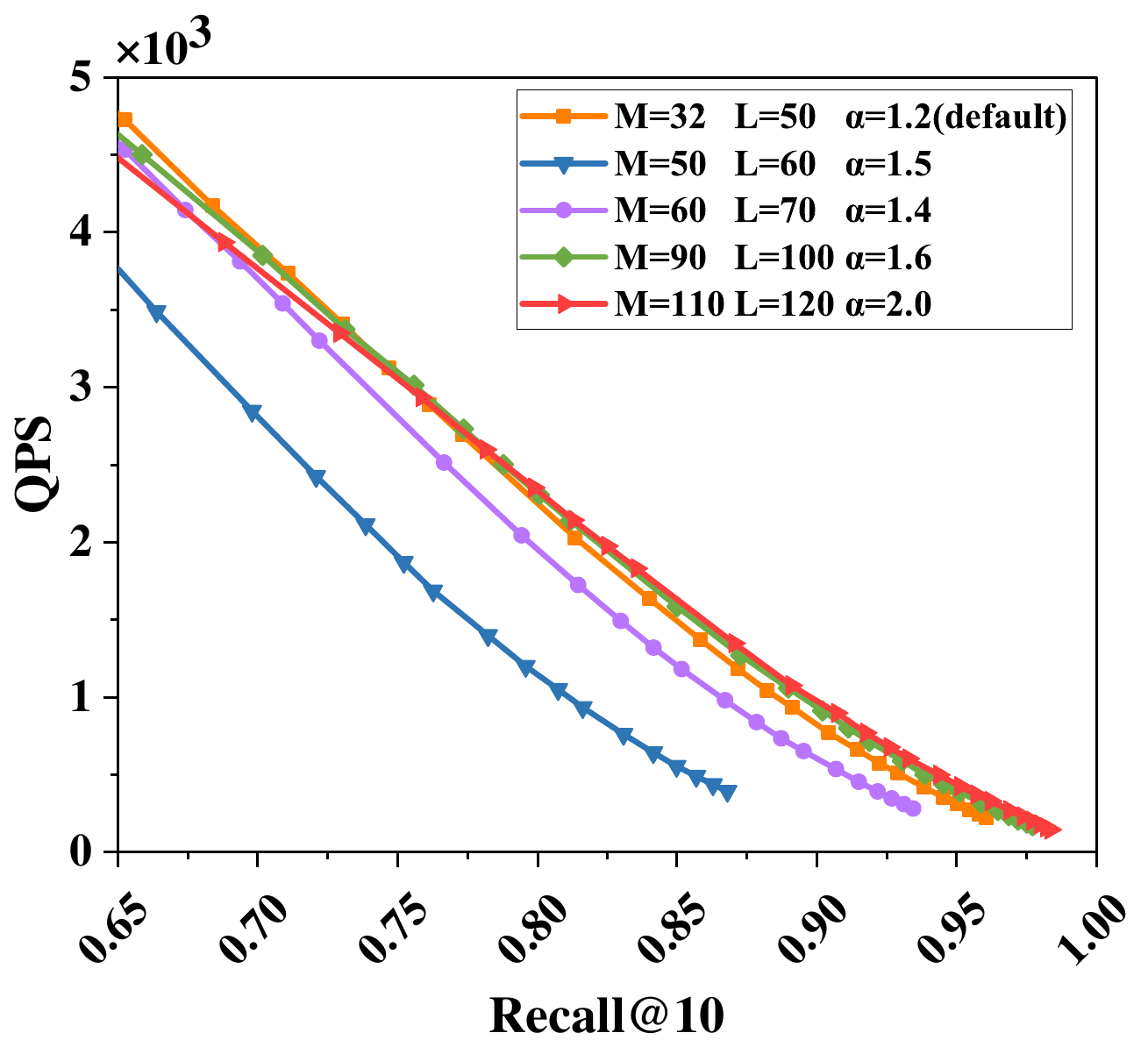}
\label{fig:cp_sensitivity-b}
}
\caption{The effects of construction parameters on $k$-ANNS performance of PGs. QPS indicates the queries processed per second, while $Recall@k$ is the accuracy of $k$-ANN obtained.}
\label{fig:cp_sensitivity}
\end{figure}

Existing tuning works for PG-based methods contain two steps: (1) parameter recommendation that recommends promising candidates of construction parameters, and (2) parameter estimation that evaluates the quality of each candidate by testing the $k$-ANNS performance of the PG built accordingly.
Existing works all focus on parameter recommendation and can be divided into two categories: heuristic-based methods and learning-based methods. The former employs heuristics to fast generate a set of parameter candidates, while the latter employs machine learning techniques to produce high-quality candidates.
Specifically, as the classic heuristic-based method, Grid Search~\cite{liashchynskyi2019grid} samples a sufficient number of candidates from the parameter space in a grid manner.
Nevertheless, because the construction cost of a PG increases superlinearly with the number of vectors in $D$, Grid Search suffers from tuning inefficiencies due to the vast number of sampled candidates, which necessitate the creation of numerous PGs.
To reduce the number of candidates, Random Search~\cite{Bergstra_Bengio_2012} randomly samples a specific number of candidates from the parameter space, improving tuning efficiency at the expense of tuning quality.
{The learning-based methods like VDTuner~\cite{yang2024vdtuner}, OtterTune \cite{van2017automatic}, and PGTuner~\cite{PGTuner} all consider expediting parameter tuning by leveraging distinct learning techniques to decrease the number of parameters recommended through a polling-and-abandonment strategy.}
However, as previously discussed, existing works tend to overlook the expenses associated with parameter estimation, which constitutes the majority of the tuning cost. For example, as illustrated later in Table~\ref{tab:time_breakdown}, in VDTuner, 98.67\% of the time cost arises from parameter estimation, i.e., the multiple PG constructions, on {Sift}.
Consequently, they all suffer from low tuning efficiency in practice.

To address this issue, we focus on accelerating parameter estimation by introducing a novel model-agnostic tuning framework. First, we explore a common manner in the parameter recommendations across various tuning methods: they generate similar sets of parameter candidates, leading to shared distance computations on similar graph structures during construction. However, storing all distance values becomes impractical due to potentially extensive memory requirements, reaching up to $O(n^2)$ for multiple PGs. To resolve this limitation, we propose a strategy of grouping PGs for parameter estimation to leverage the structural overlap among them, thereby enhancing parameter estimation efficiency. Additionally, we introduce a deterministic random strategy to further maximize the overlap among multiple PG indexes.
However, it is still non-trivial to design such a scheduling mechanism, due to two factors: (1) the SOTA recommendation models such as VDTuner~\cite{yang2024vdtuner} generate each candidate in a greedy manner, whose quality is further used to refine the model, and (2) PGs have various construction methods, making it challenging to create a method appliable for all PGs. 
For the first factor, we propose extending the SOTA method VDTuner to recommend multiple parameters per iteration, thereby building multiple PGs accordingly. 
For the second factor, we first analyze the construction process of existing PGs, and figure out two operations contributing to the major building cost, i.e,  (1) \texttt{Search} that conducts $k$-ANNS for each vector $u \in D$ on an initially built KNNG, and (2) \texttt{Prune} that prunes redundant neighbors in the $k$-ANN of each $u$. Next, by leveraging shared computations in \texttt{Search} and \texttt{Prune} phases, we design efficient multiple $k$-ANNS methods and multiple pruning methods, and then seamlessly integrate them into the construction of various PGs. 
Finally, we conducted extensive experiments on real-world datasets to demonstrate the advantages of our proposed framework, FastPGT. According to our experiments, FastPGT achieves up to 2.37x speedup over the SOTA methods in tuning efficiency while producing candidate parameters of comparable or even better quality.

Our principal contributions are as follows.
\begin{itemize}
    \item We identify the limitations of existing solutions in parameter tuning and the major challenges that need to be addressed in parameter tuning for PG index construction.
    \item We propose FastPGT, a novel PG construction tuning framework that can recommend multiple high-quality parameters and efficiently estimate their quality by simultaneously building multiple PGs.
    \item We conduct extensive experiments to validate the effectiveness of FastPGT. According to our experimental results, compared to SOTA methods, FastPGT significantly accelerates the parameter tuning for various PGs, while not compromising the tuning quality.
\end{itemize}

\section{Preliminaries}
\label{sec:preli}

Given a dataset $D \subset \mathbb{R}^d$ and a query $q \in \mathbb{R}^d$, $k$-approximate nearest neighbor ($k$-ANN) search ($k$-ANNS) returns $k$ vectors in $D$ that are sufficiently close to $q$. In this work, we use the Euclidean distance to measure the distance between two vectors $u$ and $v$ in $D$, which is denoted as $\delta(u, v)$. 

According to recent studies~\cite{AumullerBF20,LiZAH20,Dpg,survey2021, PGsurvey25}, proximity graph (PG) based methods are considered as the SOTA methods for $k$-ANNS. A PG \(G = (V, E)\) of dataset $D$ is a directed graph, where $V$ denotes its vertex set and $E$ is its edge set. Each $u \in V$ uniquely represents a vector in $D$, and an edge $(u, v) \in E$ indicates that $v$ is a close neighbor of $u$ in $D$. \(N_G(u)\) represents the set $\{v\in V \mid (u, v) \in E\}$, i.e., the out-neighbors of $u$ in $G$. During the past decade, a bulk of PGs~\cite{KGraph, Nsg, Hnsw, Dpg, diskann, Nssg, taumg, CAGRA, ALMG} have been proposed, which share the same vertex set $V$ but have distinct edge sets $E$ due to their different edge selection strategies. 

\begin{algorithm}[t]
\caption{KANNS($G, q, k, ep, ef$)}
\label{alg:knn_search}
\Input{A PG $G$, query $q$, $k$ for $k$-ANN, the entering point $ep$, and a parameter $ef$ for pool size}
\Output{$k$-ANN of $q$}
$i \gets 0$ \;
$pool[0] \gets (ep, \delta(q, ep))$\;
\While{$i < ef$}{
     $u \gets pool[i]$\;
    \For{each $v \in N_G(u)$}{
        insert $(v, \delta(q, v))$ into $pool$\;
    }
    sort $pool$ and keep the $ef$ closest neighbors\;
    $i \gets$ index of the first unexpanded point in $pool$\;
}
\Return $pool[0, \ldots, k-1]$\;
\end{algorithm}


Although various PGs have different graph structures, they share the same $k$-ANNS algorithm, as shown in Algorithm~\ref{alg:knn_search}. It starts from an entry point $ep$ (line 2), which is specified in advance or randomly selected, and greedily approaches the query $q$. A sorted array $pool$ of size $ef \geq k$, which contains the currently found closest neighbors, is maintained and takes $ep$ as its first member (line 2). The main loop (lines 3-8) iteratively extracts the closest unexpanded neighbor $u$ from $pool$ (line 4å) and then expands it (lines 5-7), where each $v \in N_G(u)$ is treated as a $k$-ANN candidate and verified by an expensive distance computation. Once all the $ef$ nodes in $pool$ have been expanded, the search terminates (line 3) and then returns the first $k$ neighbors in $pool$ as the $k$-ANN results of $q$ (line 9).

\subsection{PG Construction Methods}
\label{ssec:pg_con}

In this part, we briefly review the PG construction algorithms. First, PGs could be roughly divided into three categories: $k$-nearest neighbor graph (KNNG), relative neighborhood graph (RNG), and navigable small world graph (NSWG)~\cite{FastPG}. KNNG builds directed edges from each $u \in D$ to its $k$-ANN, but suffers from relatively poor search performance~\cite{Dpg, survey2021}. Hence, we focus on RNG and NSWG in the rest of this work, which are SOTA methods among PGs. 


\begin{figure}[t]
\centering
\subfloat[RNG construction process]{
\label{fig:illus_rng}
\includegraphics[width=\linewidth]{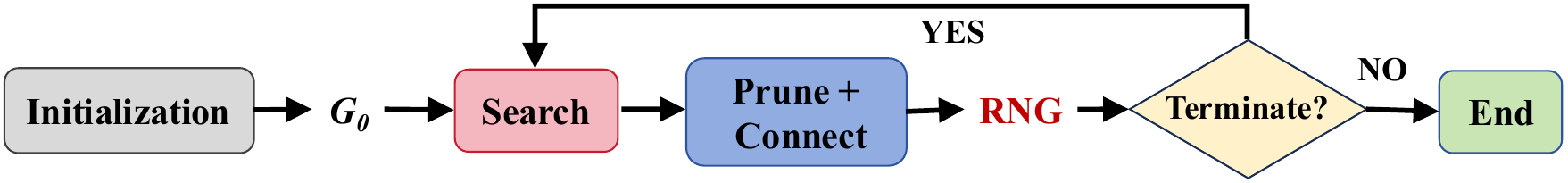}
}
\\
\subfloat[NSWG construction process]{
\label{fig:illus_nswg}
\includegraphics[width=0.8\linewidth]{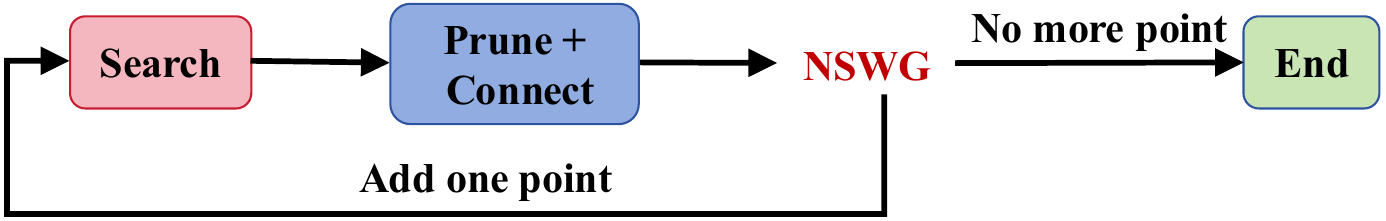}
}
\caption{The illustration of the PG construction methods.}
\label{fig:PG construction}
\end{figure}

\begin{algorithm}[t]
\small  
\SetFuncSty{textsf}
\SetArgSty{textsf}
\caption{Prune($u,C(u),M$, $\alpha$)}
\label{alg:prune}
\Input{a vertex $u$, its candidate neighbor set $C(u)$, out-degree limit $M$ and a parameter $\alpha$}
\Output{a pruned neighbor set of $u$}
$PN\leftarrow \emptyset$\;
\For{each $v \in C(u)$ in the ascending order of $\delta(u,v)$}
{
    $Flag\leftarrow$ false\;
    \For{each $w \in PN$}
    {
        \If{$\alpha \cdot\delta(v,w)<\delta(u,v)$}
        {
            $Flag\leftarrow$ true\;
        }
    }
    \If{$Flag=$ false} {
        $PN\leftarrow PN\cup \{v\}$\;
    }
    \textbf{if} $|PN|\ge M$ \textbf{then break}\;
}
\Return{$PN$}\;
\end{algorithm}

RNG methods such as NSG~\cite{Nsg}, DPG~\cite{Dpg}, NSSG~\cite{Nssg} and Vamana~\cite{diskann} assume that they have the full knowledge $D$ in advance. As illustrated in Figure~\ref{fig:illus_rng}, their construction could be divided into four operations, i.e., \texttt{Initialization}, \texttt{Search}, \texttt{Prune} and \texttt{Connect}. 
\texttt{Initialization} builds an initial PG $G_0$, which could be a KNNG with moderate accuracy as in NSG~\cite{Nsg, Dpg, Nssg, CAGRA} or a randomly generated KNNG as in Vamana~\cite{diskann}. 
Then, \texttt{Search} conducts $k$-ANNS (Algorithm~\ref{alg:knn_search}) for each $u \in D$ with $G_0$ in order to obtain refined $k$-ANN results, denoted as $C(u)$. 
In \texttt{Prune}, each $C(u)$ is pruned to generate $N_G(u)$. The most widely used pruning method is the RNG pruning, where $(u, v) \in E$ iff there exists no edge $(u, w) \in E$ such that $\delta(u, w) < \delta(u, v)$ and  $\delta(v, w) < \alpha \cdot \delta(u, v)$. 
$\alpha$ is set as 1 in NSG and HNSW, but is a configurable parameter in Vamana. 
To be specific, we present the RNG pruning in Algorithm~\ref{alg:prune}. 
The pruned neighbor set $PN$ for $C(u)$ that is sorted in the ascending order of the distance to $u$, is initialized as $\emptyset$ (line 1). Then, each neighbor $v \in C(u)$ is tried to insert into $PN$ (lines 2-9). $Flag$ is initialized as false (line 3), which indicates whether or not a neighbor in $PN$ dominates $v$. If $\delta(v, w) < \alpha \cdot \delta(u, v)$ (line 5), we say that $w \in PN$ dominates $v$, and set $Flag$ as true (line 6). $v$ will be inserted into $PN$ if all neighbors in $PN$ do not dominate $v$ (lines 7-8). Once $PN$ contains $M$ members, the pruning operation terminates (line 9). 
As to \texttt{Connect}, it aims to enhance the connectivity of the built PG by adding reversing edges or additional edges between connected components in the current PG. 

Unlike RNG, NSWG methods such as NSW~\cite{Sw} and HNSW~\cite{Hnsw} build the PG from scratch and insert each $u \in D$ into the current incomplete NSWG index one by one without the full knowledge of $D$. As depicted in Figure~\ref{fig:illus_nswg}, the insertion of each $u$ first conducts $k$-ANNS on the currently incomplete NSWG in \texttt{Search}. Next, they conduct pruning of the $k$-ANNS results to obtain the final neighbor set by the RNG pruning in HNSW~\cite{Hnsw} (Algorithm~\ref{alg:prune}) or directly removing some neighbors with larger distances in NSW~\cite{Sw}. Regarding \texttt{Connect}, NSWG adds reverse edges to enhance graph connectivity. 

\subsection{Problem Statement}
\label{ssec:prob}

In this paper, we focus on efficiently and effectively tuning the PG construction parameters of both NSWG and RNG methods for their optimal search performance.
For the sake of discussion, we take HNSW~\cite{Hnsw} as the representative NSWG method, and Vamana~\cite{diskann}, NSG~\cite{Nsg} as the representative RNG methods. Notably, HNSW and NSG are two SOTA in-memory solutions, while Vamana is the key component of DiskANN~\cite{diskann}, which is the SOTA I/O-efficient PG method.

Let $N(q)$ denote the set of $k$-ANNS results, and $N^{*}(q)$ the ground truth, $Recall@k = {|N(q) \cap N^{*}(q)|}/{k}$.
Formally, we give the problem statement as follows.
Given an integer $k$ of returned neighbors and a target $Recall@k$ value $t$, our objective is to tune the construction parameters of PG and the $k$-ANNS parameter $ef$ (as in Algorithm~\ref{alg:knn_search}) to achieve the $k$-ANNS performance of $Recall@k = t$. 
As aforementioned, we focus on three PGs in this work, i.e., HNSW, Vamana and NSG. 
\textbf{(i) HNSW} has two construction parameters, i.e., (1) $efc$ indicates the parameter $ef$ of $k$-ANNS for each $u \in D$ in \texttt{Search}, and (2) $M$ indicates the out-degree limit in \texttt{Prune}. 
\textbf{(ii) Vamana} requires four construction parameters, i.e., (1) $L$ indicates $ef$ in $k$-ANNS during \texttt{Search}, (2) $R$ is the number of returned neighbors in \texttt{Search}, (3) $M$ is the out-degree limit in \texttt{Prune}, and (4) $\alpha$ is a parameter in \texttt{Prune}, i.e., $\alpha$ of Algorithm~\ref{alg:prune}.
\textbf{(iii) NSG} needs four construction parameters, i.e., (1) $K$ indicates the out-degree of the initial graph in \texttt{Initialization}, (2-4) $L,R,M$ has the same meaning as in Vamana.
Notably, our method can be easily extended to other PGs, given their inherent similarity in construction~\cite{FastPG}.

\section{Existing Methods and Their Issues}

In this section, we first review the existing methods and then point out their issues.
In the literature, there exists a bulk of methods that can be used to tune the parameters for PG index construction. 
In general, existing methods contain two essential operations: (1) parameter recommendation (denoted as {Recom.}) and (2) parameter estimation (denoted as Est.). The former selects and recommends promising parameter candidates, while the latter estimates the quality of each recommended candidate. 

Existing works all focus on efficiently recommending parameter candidates via various methods, which could be divided into two categories: heuristic-based methods~\cite{liashchynskyi2019grid,Bergstra_Bengio_2012} and learning-based methods~\cite{van2017automatic,yang2024vdtuner,PGTuner}.  
%
As the widely-used heuristic-based method, Grid Search~\cite{liashchynskyi2019grid} exhaustively enumerates discrete parameter combinations, at the expense of high computational costs with numerous candidates. To enhance efficiency, Random Search~\cite{Bergstra_Bengio_2012} samples parameters randomly in the parameter space at the risk of missing the parameters of high quality. 
%
%
Learning-based methods construct predictive models to understand the relationships between parameters and performance metrics, and then use them to guide parameter recommendations. OtterTune \cite{van2017automatic} employs Gaussian Process Regression to identify optimal parameters. 
{PGTuner~\cite{PGTuner} leverages a pre-trained query performance model and a deep reinforcement learning based parameter recommendation model to produce high-quality candidates. 
VDTuner \cite{yang2024vdtuner}, as a state-of-the-art method, leverages multi-objective Bayesian optimization and employs Expected Hypervolume Improvement (EHVI) as the metric to find optimal configurations.
Notably, VDTuner and PGTuner are specially designed for the PG construction tuning, while others are general-purpose tuning methods. 
Moreover, VDTuner assumes no prior knowledge of the tuning task, whereas PGTuner relies on a pre-trained model to accelerate it.
In this paper, we follow the same assumption as VDTuner that no prior knowledge is required and tune the PG from scratch. } 

\vspace{1mm}
\stitle{Issues of Existing Methods.}
All existing works focus on recommending high-quality candidates but overlook the parameter estimation, which is the main contributor to the total tuning cost. 
To estimate the quality of each recommended parameter, existing methods directly build the corresponding PG and then test its $k$-ANNS performance. 
Notably, the cost of building a PG is superlinear to the number of vectors in $D$, which is pretty costly~\cite{FastPG}. 
As shown in Table~\ref{tab:time_breakdown}, at least 98.5\% of the total cost comes from parameter estimation.
Hence, it is vital to accelerate the parameter estimation in order to improve the overall tuning efficiency. 


\begin{table}[t]
\centering
\caption{The Cost Decomposition of Existing Methods on Sift.}
\label{tab:time_breakdown}
\begin{tabular}{|c|r|c|c|}
\hline
\textbf{Method} & \textbf{{Recom.}} & \textbf{{Est.}} & \textbf{Total} \\ \hline\hline
RandomSearch~\cite{Bergstra_Bengio_2012}  &     2s  & 40,570s~(100.00\%) & 40,572s      \\ \hline
VDTuner~\cite{yang2024vdtuner} & 438s  & 32,568s~~(98.67\%) & 33,006s \\ \hline
OtterTune~\cite{van2017automatic} & 161s & 39,730s~~(99.59\%) & 39,891s \\ \hline
FastPGT (ours) & 617s & 14,428s~~(95.89\%) & 15,045s \\ \hline
\end{tabular}
\end{table}

\section{Our Solution}
\label{sec:sln}

In this section, we introduce our FastPGT method, focusing on expediting the parameter estimation operation, a key contributor to overall cost but is overlooked by existing methods.
To begin, in Section~\ref{ssec:paramter-r}, we analyze the parameter $R$ utilized in RNG construction and propose its removal from the parameter space to streamline parameter enumeration and enhance acceleration. Next, to leverage computing resources effectively, we advocate recommending multiple parameters and concurrently constructing multiple PGs with respect to various parameter candidates in Section~\ref{ssec:parallel_pr}. 
Furthermore, we examine the time breakdown of current PG-based algorithms, emphasizing the time-intensive aspects of \texttt{Search} and \texttt{Prune} to expedite parameter estimation, as detailed in Section~\ref{ssec:cost_analysis}. We then propose an efficient algorithm for batch parameter estimation that facilitates the construction of multiple PGs simultaneously in Sections~\ref{ssec:Shared Distance}, \ref{ssec:Shared Reusing Prune}, and \ref{ssec:m_pg}.

\begin{figure*}[t]
\centering
\includegraphics[width=0.8\linewidth]{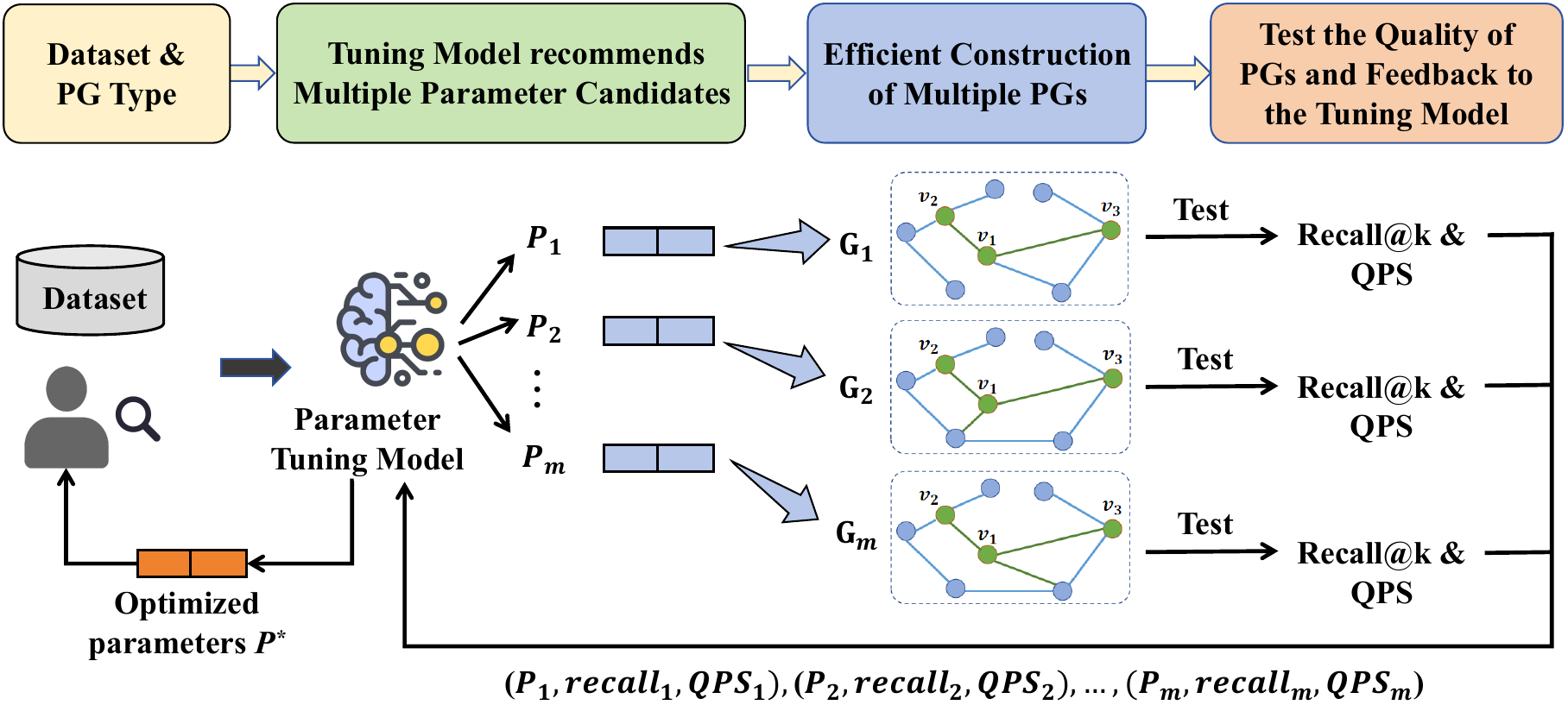}
\caption{Illustrating the framework of our method FastPGT.}
\label{fig:framework}
\end{figure*}

\vspace{1mm}
\stitle{Overview}: We depict the pipeline of FastPGT in Figure~\ref{fig:framework}, where iterative parameter tuning takes place. Initially, we utilize a parameter recommendation model to recommend multiple promising construction parameters.
To expedite the construction of multiple PGs in line with the recommended parameters, we exploit the common graph structures and the shared distance computations during index construction. After testing the constructed PGs, the tuning model will be further refined using the parameters and their corresponding performance metrics, and new parameters will be proposed for evaluation.

It is worth noting that our FastPGT is model-agnostic and can accelerate parameter estimation across various existing parameter recommendation methods. In this paper, we use VDTuner as the tuning model for parameter recommendations. Given that VDTuner lacks inherent support for recommending multiple parameters concurrently, we extend it to enable this capability.

\subsection{Parameter $R$ in RNG Construction}
\label{ssec:paramter-r}

In this part, we analyze the parameter $R$ within RNG construction, as utilized in both NSG and Vamana. This parameter controls the number of neighbors returned in the \texttt{Search} phase for subsequent \texttt{Prune}.
To elaborate, consider $pool_L=\{v_1,v_2,\cdots, v_L\}$ as the data point set obtained from $k$-ANNS (Algorithm~\ref{alg:knn_search}) when $ef=L$, with $q$ is a data point $u\in D$. Assuming $\delta(v_1,u)\le\delta(v_2,u)\le\cdots\le\delta(v_L,u)$, $R$ is employed to establish $C(u)=\{v_1,v_2,\cdots,v_R\}$, where $R\le L$.
Let $PN(R)$ denote the set $PN$ derived from the \texttt{Prune} process (Algorithm~\ref{alg:prune}) when $|C(u)|=R$ for a node $u$ and $C(u)$ is conducted with the same value of $L$ and same PG. The subsequent theorem demonstrates that the set $PN(R)$ is a subset of $PN(R')$ when $R\le R'$.

\begin{theorem}
Given a data point $u\in D$ and a set $\{v_1,v_2,\cdots,v_L\}$ with $\delta(v_1,u)\le\delta(v_2,u)\le\cdots\le\delta(v_L,u)$, let $PN(R)$ be the results produced by Algorithm~\ref{alg:prune} with $C(u)=\{v_1,\cdots,v_R\}$.
For the same $M$ and $\alpha$ values, it follows that $PN(R)\subseteq PN(R')$ when $R\le R'\le L$.
\end{theorem}

\begin{proof}
If $|PN(R)|\ge M$, the algorithm returns $PN$ upon termination at line 9. Hence, elements $v\in C(u)$ with greater distances to $u$ are not enumerated. This results in $PN(R)=PN(R')$ for $R'\ge R$.
Otherwise, if we let $PN(R)={v_{p_1},v_{p_2},\cdots,v_{p_c}}$, where $c< M$ and $p_1<p_2<\cdots<p_c$, for $i<p_c$ such that $v_i\not\in PN(R)$, it follows that $v_i\not\in PN(R')$ as well. This is because both $PN(R)$ and $PN(R')$ operate on the same elements in the same order for pruning nodes, if a node $v_j\in PN(R)$ satisfies $\alpha \cdot \delta(v_i,v_j)<\delta(v_i,u)$, it will also exclude $v_i$ from $PN(R')$.
\end{proof}

Hence, from the above theorem, it is evident that increasing the value of $R$ initially results in a sufficient neighbor list, i.e., the $|PN|$ increases and finally reaches $M$, and then will not affect $PN$. Additionally, since the efficiency of Algorithm~\ref{alg:prune} remains unaffected by the value of $R$, courtesy of the early break at line 9, the traversal of more candidate nodes in $C(u)$ is avoided when the neighbor size is $M$. Therefore, setting the value of $R$ as $L$ during index construction is a direct approach, rendering fine-tuning unnecessary, as opting for $R=L$ proves to be the optimal choice.

\subsection{Multiple Parameter Recommendations}
\label{ssec:parallel_pr}


To maximize the utilization of computing resources, we aim for the recommendation model to generate multiple promising parameters. This approach enables the simultaneous generation of multiple PGs, facilitating the assessment of their performance for subsequent refinement of the recommendation model and parameter recommendations.

In this part, we consider extending the SOTA method VDTuner~\cite{yang2024vdtuner} to yield multiple parameters in each iteration.
VDTuner utilizes Multi-Objective Bayesian Optimization (MOBO), which is designed to balance the tradeoff between two key metrics: QPS and $Recall@k$ in $k$-ANNS performance.

VDTuner consists of two primary components: a surrogate model and an acquisition function. The surrogate model maps each construction parameter to the $k$-ANNS performance of the PG built accordingly, while the acquisition function utilizes the surrogate model's predictions to balance ``exploration'' and ``exploitation'' in the next recommendation step. Specifically, VDTuner employs a Gaussian process (GP) as the surrogate model to capture the intricate relationships between the PG parameters, i.e., $M,efc,ef$ of NSWG and $L,M,\alpha,ef$ of RNG, and the performance metrics, i.e., QPS and $Recall@k$.
To stabilize the GP training when performance values vary significantly, the model is trained on normalized performance indicators. For a parameter $P_i$ with performance metrics $(y_{i}^{\text{qps}},y_{i}^{\text{recall}})$, where $y_{i}^{\text{qps}}$ and $y_{i}^{\text{recall}}$ represents its QPS and $Recall@k$ respectively, its normalized performance $(\hat{y}_{i}^{\text{qps}},\hat{y}_{i}^{\text{recall}})$ is defined as
\begin{equation}
\label{eq:pmetric}
(\hat{y}_{i}^{\text{qps}},\hat{y}_{i}^{\text{recall}})=(\frac{y_{i}^{\text{qps}}}{\overline{y}^{\text{qps}}},\frac{y_{i}^{\text{recall}}}{\overline{y}^{\text{recall}}}),
\end{equation}
where $(\overline{y}^{\text{qps}},\overline{y}^{\text{recall}})$ is calculated over the set of all currently found non-dominated parameters representing the most balanced solution found so far, i.e., $$\arg\max_{(y^{\text{qps}},y^{\text{recall}})\in\mathcal{Y}}\frac{1}{|y^{\text{qps}}/y_{\max}^{\text{qps}}-y^{\text{recall}}/y_{\max}^{\text{recall}}|},$$
%
%
where $y_{\max}^{\text{qps}}$ and $y_{\max}^{\text{recall}}$ are the maximum QPS and $Recall@k$ found within the non-dominated set $\mathcal{Y}$. 

VDTuner employs Expected Hypervolume (HV) Improvement (EHVI) as its acquisition function~\cite{Yang_Emmerich_Deutz_Bäck_2019}. HV measures the volume dominated by a set of non-dominated solutions, i.e., the Pareto front, with a lower bound defined by a preset reference point $r$. The EHVI acquisition function directs the search by identifying the single candidate point that maximizes the EHVI.
However, the standard EHVI is primarily designed for \textit{sequential} optimization, recommending a single best candidate per iteration~\cite{Yang_Emmerich_Deutz_Bäck_2019}. Moreover, it is non-trivial to extend it for batch recommendations, since no analytical formula exists for computing the joint expected EHVI of multiple candidates. To address this issue, we propose a heuristic method, called $m$EHVI, that enables batch recommendations, where $m$ is the number of parameters recommended in each batch. 
%
The $m$EHVI computes the exact joint hypervolume improvement (HVIs) from $m$ candidates to model their collective effect. Given a set $\mathcal{Y}$ of already evaluated non-dominated points and a reference point $r$, the $m$EHVI for a candidate set $P = \{P_1, P_2, \ldots, P_m\}$ is defined as:
\begin{align}
\alpha_{m\text{EHVI}}\left(\left\{P_{i}\right\}_{i=1}^{m}\right)    =  \mathbb{E}\left[ \text{HVI}(\{f(P_i)\}_{i=1}^m, \mathcal{Y}, r) \right]  \nonumber \\
\label{eq:qehvi}
 = \mathbb{E}\left[ HV(r, \mathcal{Y} \cup \{f(P_i)\}_{i=1}^m) - HV(r, \mathcal{Y}) \right],
\end{align}
where $f(P_i)$ denotes the normalized performance of the new candidate $P_i$ predicted by the surrogate model. 
The HV function is utilized to calculate the HV of the observed parameter-performance pairs. Thus, $\alpha_{m\text{EHVI}}$ quantifies the joined EHVI after adding the $m$ candidates.

\subsection{Parameter Estimation Analysis}
\label{ssec:cost_analysis}
\begin{figure}[t]
\centering
\subfloat[Sift]{\includegraphics[width=0.47\linewidth]{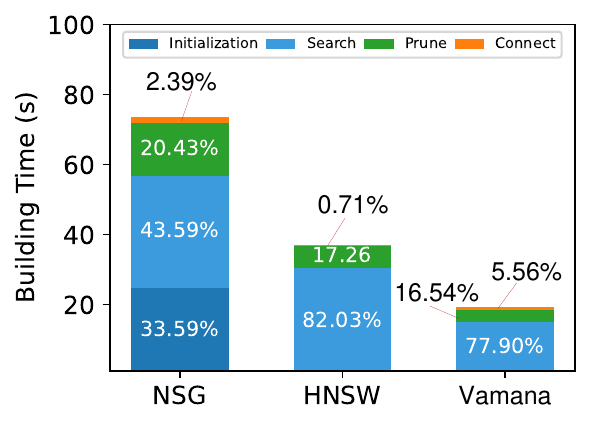}}
\hspace{0.1cm}
\subfloat[Gist]{\includegraphics[width=0.47\linewidth]{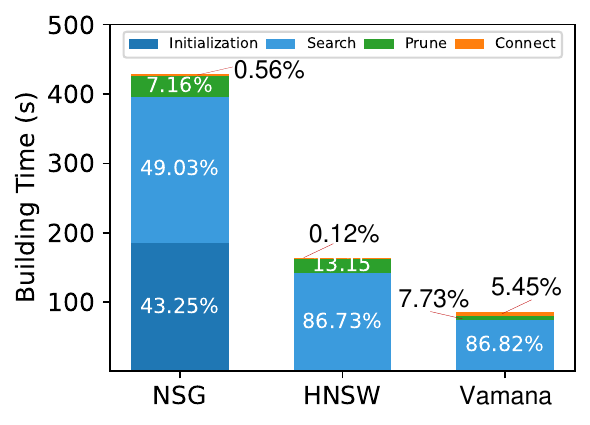}}
\caption{The cost decomposition of PG construction.}
\label{fig:cost decomposition of PG}
\end{figure}

In this part, we first analyze the construction cost of the three representative PGs: HNSW~\cite{hnswlib}, NSG~\cite{Nsg}, and Vamana~\cite{diskann}. 
As shown in Figure~\ref{fig:cost decomposition of PG}, we can find that \texttt{Search} contributes the most to the overall construction cost in all three PGs. In particular, \texttt{Search} occupies 86.7\%, 86.8\%, and 49.0\% of HNSW, Vamana, and NSG on Gist, respectively. 
Moreover, \texttt{Prune} constitutes a significant portion of the total cost.
Therefore, it is crucial to improve parameter estimation efficiency, accelerating \texttt{Search} and \texttt{Prune}.

Next, we analyze the graph structures of PGs under various parameters and demonstrate significant overlaps among graph structures, i.e., their neighbor lists. We first consider the parameter $M$ used in all three representative PGs. The following theorem demonstrates that the neighbor list for a smaller value of $M$ is a subset of that for a larger value of $M$.

\begin{theorem}
Given a data point $u\in D$ and a set $C(u)$, let $PN(M)$ be the results produced by Algorithm~\ref{alg:prune} with input $M$.
For the same $\alpha$ value, it follows that $PN(M)\subseteq PN(M')$ when $M\le M'\le |C(u)|$.
\end{theorem}

\begin{proof}
If $|PN(M)| < M$, all elements $v\in C(u)$ have been traversed. Thus, the value of $PN$ is independent of $M$, leading to $PN(M)=PN(M')$. 
If $|PN(M)|\ge M$, the algorithm returns $PN$ at line 9 upon termination. Hence, elements $v\in C(u)$ with greater distances to $u$ are not enumerated for $M$, but they may be included with a larger value of $M$ and could be contained in $PN(M')$. This implies $PN(M)\subseteq PN(M')$ for $M'\ge M$.
\end{proof}

Next, we analyze the size $|C(u)|$ of the candidate neighbor set for \texttt{Prune}, i.e., $efc$ in HNSW and $L$ in NSG and Vamana, as well as $\alpha$. We analyze these parameters in an experimental study to demonstrate their effects on the overlaps between the graph structures. Here, we take Vamana as an example. For two PGs $G_1$ and $G_2$ with different parameters, the overlap of two neighbor lists on $u$, i.e., $N_{G_1}(u)$ and $N_{G_2}(u)$, is defined as $NLO(N_{G_1}(u), N_{G_2}(u)) = |N_{G_1}(u) \cap N_{G_2}(u)| / |N_{G_1}(u)|$. Hence, the overlap between $G_1$ and $G_2$ is defined as $NLO(G_1, G_2) = 1/|D| \cdot \sum_{u\in D} NLO(N_{G_1}(u), N_{G_2}(u))$. We show the results in Figure~\ref{fig:overlap_analysis}, where we fix $M=50$. In Figure~\ref {fig:overlap_analysis-a}, we fix $\alpha=1.2$ and present the $NLO$ value between two graphs with different $L$ values. Obviously, the closer $L$ values are, the larger the $NLO$ is, i.e., the more similar their graph structures are. In Figure~\ref{fig:overlap_analysis-b}, we fix $L = 100$ and vary $\alpha$. Similarly, the closer $\alpha$ values lead to similar graph structures built accordingly. Besides, similar phenomena could be observed in other PGs and other datasets, which are omitted due to the space limit. 

\begin{figure}[t]
\centering
\subfloat[Overlap for $L$]{\includegraphics[width=0.46\linewidth]{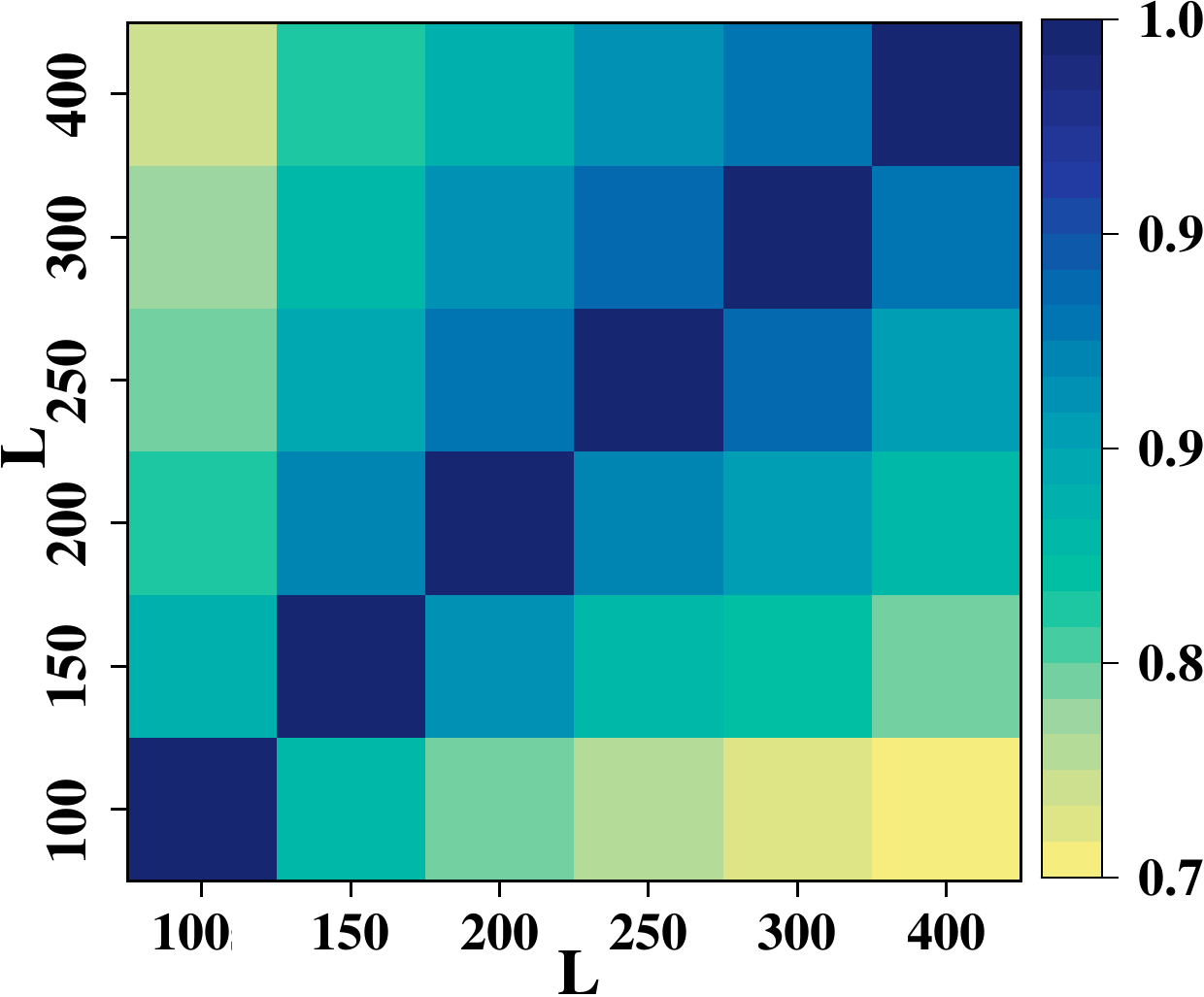}
\label{fig:overlap_analysis-a}
}
\hspace{0.1cm}
\subfloat[Overlap for $\alpha$]{\includegraphics[width=0.46\linewidth]{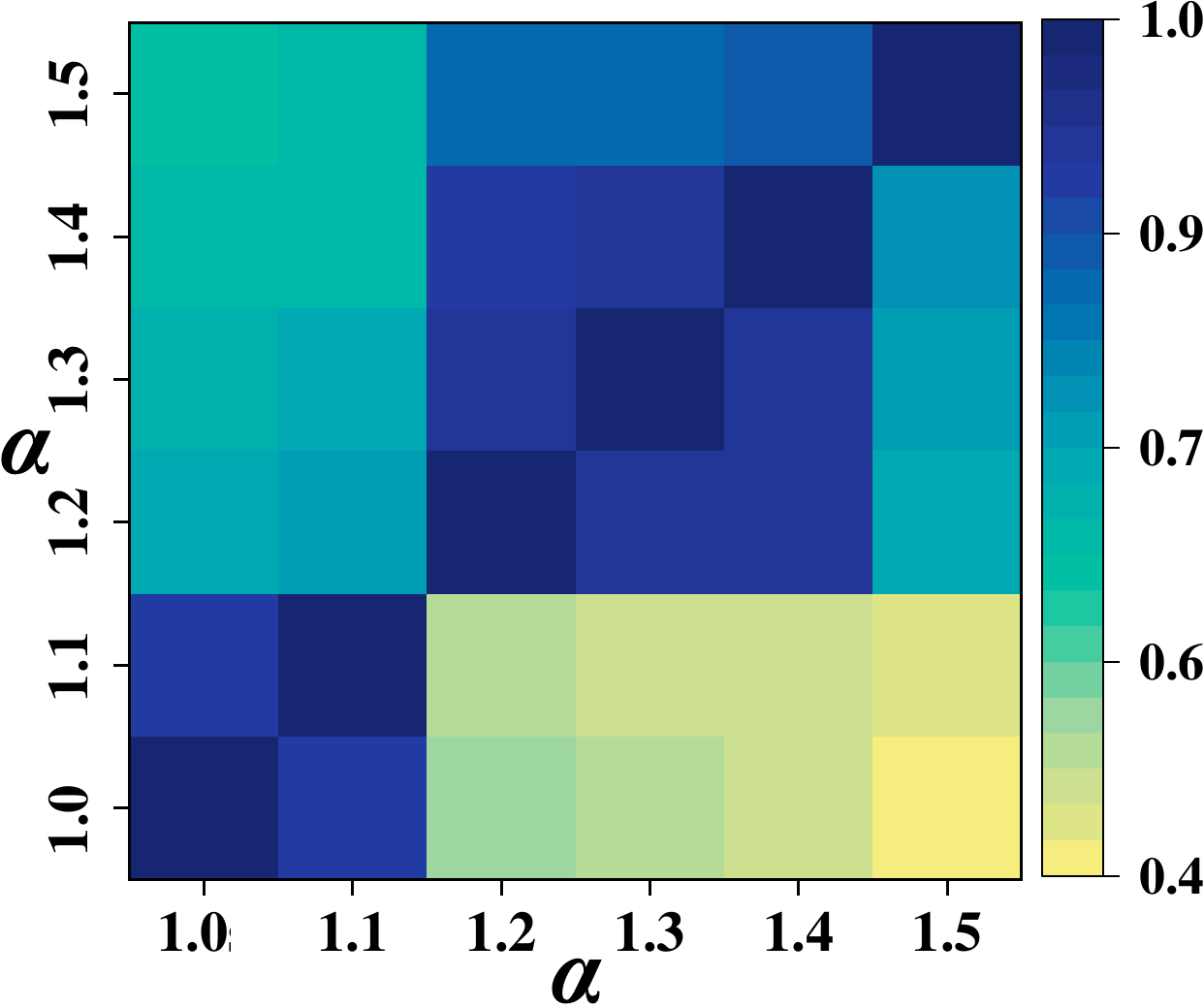}
\label{fig:overlap_analysis-b}
}
\caption{The effects of $L$ and $\alpha$ on the neighbor list overlap between the Vamana graphs built accordingly on Sift. }
\label{fig:overlap_analysis}
\end{figure}

As discussed above, we observe significant overlaps in the graph structures of the PGs built from similar parameters. Hence, we explore the efficient simultaneous construction of multiple PGs to leverage these overlaps. This approach enables us to maximize the utilization of the common graph structures and similar computations during construction.

To further enhance the overlap among multiple PGs, we employ a deterministic random strategy for HNSW layer determination and KNNG initialization in NSG and Vamana.

\vspace{1mm}
\stitle{Deterministic Random Strategy}: 
The key concept behind deterministic randomness is to leverage pseudo-random number generators to deterministically produce the same number sequence with a consistent seed value. In detail, for each node in HNSW, we ascertain its layer deterministically prior to subsequent index construction. This ensures uniformity across multiple PGs in terms of node layers. 
In NSG and Vamana, to build the initial KNNG $G_0$, we deterministically establish the initial neighbor list for each node. These neighbor lists are then utilized to initialize KNNG across all concurrently constructed PGs during parameter estimation.
To reduce memory usage, we opt not to store neighbor lists and layer information explicitly. Instead, we generate them using a pseudo-random number generator with a consistent seed value set across all instances.

\subsection{Efficient \texttt{Search} Operations for Multiple PGs}
\label{ssec:Shared Distance}


In this part, we aim to accelerate the \texttt{Search} operations when building multiple PGs simultaneously. Notably, the \texttt{Search} operation independently conducts $k$-ANNS on the current PG (e.g., the current HNSW, the initially built KNNG, or the current Vamana) for each $u \in D$. Hence, we consider conducting the $k$-ANNS for $u \in D$ on the multiple PGs simultaneously, which does not affect the results of the \texttt{Search} operations. 
As follows, we first analyze the shared computations among a node $u\in D$ in multiple PGs.

\vspace{1mm}
\stitle{Repeated Computations in \texttt{Search} Operations:}
As shared graph structures and other overlaps lead to repeated distance computations, which constitute a significant cost in $k$-ANNS, as shown in prior studies~\cite{FastPG}, our focus here is on the distance computations across multiple PGs.
There exist numerous distance computations when building multiple PGs, especially in the \texttt{Search} operations. 
Intuitively, this is because the distances of $u \in D$  to its close points will be computed more than once when conducting $k$-ANNS for $u$ on multiple PGs. Moreover, the more similar their construction parameters are, the more repeated distance computations. 
The more \texttt{Search} operations conducted simultaneously, the more repeated distance computations are. 


\begin{table}[t]
\setlength{\abovecaptionskip}{0cm}
\caption{Illustrating repeated computations on Sift and Glove under different parameter settings for HNSW.}
\label{tb:repeation}
\begin{center}
  \resizebox{0.49\textwidth}{!}{
\begin{tabular}{|c|c|r|c|c|c|}
\hline
\textbf{Datasets}&
\text{$(efc,M)$} & {$\#dist$ ($\times 10^9$)} & \text{$ratio_{rp}$} & \text{$ratio_{rp}^{s}$}&
\text{$ratio_{rp}^{p}$} \\
\hline \hline
\multirow{7}{*}{Sift}&\text{A:(300,18)}& 5   &   &&  \\ 
&\text{B:(300,20)}& 5.3 &  54\% & 60\% & 33\%\\ 
&\text{C:(300,22)}& 5.5 &   &  &\\  
\cline{2-6}
&\text{A:(400,28)}& 7.7 &   & &\\
&\text{B:(400,30)}& 7.9 &  56\% & 61\% & 36\% \\
&\text{C:(400,32)}& 8.1 &   & & \\
\hline
\multirow{7}{*}{Glove}
&\text{A:(300,18)}& 7.5   &   &  &\\ 
&\text{B:(300,20)}& 7.9 &  55\% & 60\% & 35\%\\ 
&\text{C:(300,22)}& 8.3 &   &  &\\  
\cline{2-6}
&\text{A:(400,28)}& 12.1 &   & &\\
&\text{B:(400,30)}& 12.5 &  58\% & 63\% & 38\% \\
&\text{C:(400,32)}& 12.9 &   &  &\\
\hline
\end{tabular}
}
\end{center}
\end{table}

To illustrate this phenomenon, we take HNSW as an example and conduct experiments in the real-world datasets, whose results are shown in Table \ref{tb:repeation}. Here, we build three HNSW graphs with three construction parameters and count the repeated distance computations in the construction of those HNSW graphs.
Let $dist_A$, $dist_B$, and $dist_C$ be the set of vector pairs for distance computations under the parameter settings $A$, $B$, and $C$, respectively. We use $\#dist$ to denote the total distance computations, and $ratio_{rp}$ to denote the ratio of common distance computations among all three settings, i.e.,
$ratio_{rp} = |dist_A \cap dist_B \cap dist_C| / (|dist_A| + |dist_B| + |dist_C|)$. 
Let $dist_A^{s}$, $dist_B^{s}$, and $dist_C^{s}$ be the set of vector pairs for distance computations caused by \texttt{Search} operations. $ratio_{rp}^{s} = |dist_A^{s} \cap dist_B^{s} \cap dist_C^{s}| / (|dist_A^{s}| + |dist_B^{s}| + |dist_C^{s}|)$. 
The results show that more than half of the distance computations are repeated under the different parameter settings. 
Regarding \texttt{Search} operations, they share a substantial portion (i.e., at least 60\%) of redundant distance computations, excluding those that could be prevented by utilizing a bitmap $visited$ in Algorithm \ref{alg:knn_search}, where each element in $visit$ indicates whether or not it has been verified in this $k$-ANNS operation.

\begin{algorithm}[t]
\caption{mKANNS($G, q, k, ef, ep,V_{\delta}$)}
\label{alg:mkanns}
\Input{A PG $G$, query $q$, $k$ for $k$-ANN, the entering point $ep$, a parameter $ef$ for pool size, and $V_{\delta}$ for record distance computations}
\Output{$k$-ANN of query $u$}
$i \gets 0$\;
$pool[0] \gets (ep, \text{dist}(u, ep))$\;
\While{$i < ef$} {
    $u \gets pool[i]$\;
    \For{each $v \in N_G(u)$} {
        \If{{$V_{\delta}[v] != -1$}} {
           { fetch $\delta(u, v)$ from $V_{\delta}[v]$};
        }
        \Else{
            {compute $\delta(u, v)$ and store it in $V_{\delta}[v]$};
        }
        insert $(v, \delta(q, v))$ into $pool$;
    }
    sort $pool$ and keep the $ef$ closest neighbors\;
    $i \gets$ index of the first unexpanded point in $pool$\;
}
\Return{$pool[0, \ldots, k-1]$}
\end{algorithm}

\vspace{1mm}
\stitle{Fast Multiple \texttt{Search} Operations:}
Suppose that we build $m$ PGs $\{G_1, G_2, \ldots, G_m\}$ simultaneously w.r.t. a set of parameters $P = \{P_1, P_2, \ldots, P_m\}$. For each node $u \in D$ and each parameter $P_i$ ($1 \leq i \leq m$), we need to conduct a $k$-ANNS on a specific graph index. To avoid the repeated distance computations during those $m$ $k$-ANNS operations for $u$, we cache all the distances computed so far. Let $V_{\delta}$ with $n$ elements be the set of cached distances, where each element is initialized as -1. Once $\delta(u, v)$ is computed, $V_{\delta}[v] = \delta(u, v)$.

We show the $k$-ANNS operation with $V_{\delta}$ in Algorithm \ref{alg:mkanns},
which is similar to Algorithm \ref{alg:knn_search} but differs in the distance computations. Before each distance $\delta(u, v)$, Algorithm \ref{alg:mkanns} checks $V_{\delta}[v]$  first to determine whether it has been cached (line 6). If $\delta(u, v)$ has been computed, just fetch it from $V_{\delta}$ (line 7). Otherwise, directly compute it and update $V_{\delta}[v]$. In this way, our method avoids repeated distance computations across multiple $k$-ANNS operations on the same $u$. 

\subsection{Efficient \texttt{Prune} Operations for Multiple PGs}
\label{ssec:Shared Reusing Prune}

In this part, let us consider $m$ $k$-ANNS for the same node $u \in D$ on $m$ distinct PGs, as discussed in the last part. Let $C_1(u), C_2(u), \ldots, C_m(u)$ be the corresponding $m$ sets of $k$-ANNS results. Then, we need to prune each $C_i(u)$ ($1 \leq i \leq m$) in each \texttt{Pune} operation to get the pruned neighbor set $C_i^{'}(u)$. 
Notably, as in Figure~\ref{fig:cost decomposition of PG}, \texttt{Prune} operations make a significant contribution to the overall construction cost. Hence, accelerating \texttt{Prune} operations enhances overall construction efficiency. 
Like multiple \texttt{Search} operations on the same $u$, we observe the repeated computations among those $m$ \texttt{Prune} operations. 

\begin{figure}[t]
\centering
\includegraphics[width=1\linewidth]{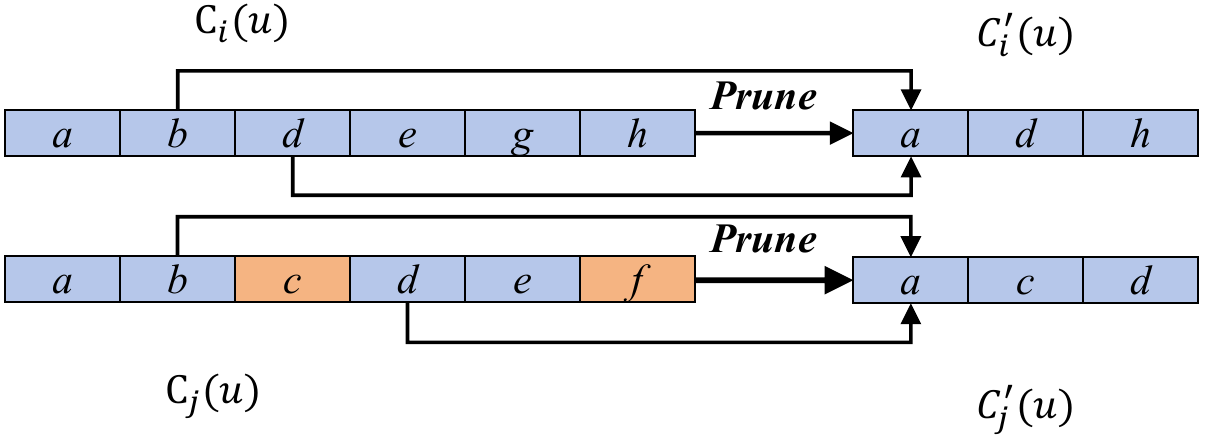}
\caption{Illustrating the repeated computations among two \texttt{Prune} operations on the same node $u$ under different parameters. Here, $\delta(a, d)$ is computed in both \texttt{Prune} operations.}
\label{fig:repeation_in_prune}
\end{figure}

\vspace{1mm}
\stitle{Repeated Computations in Multiple \texttt{Prune} Operations.}
Consider $C_i(u), C_j(u)$, where $1 \leq i \neq j \leq m$. Since they both contain the close neighbors of $u$, it is easy to derive that $C_i(u) \cap C_j(u) \neq \emptyset$. Moreover, the more similar the corresponding construction parameters $P_i$ and $P_j$ are, the more similar $C_i(u)$ and $C_j(u)$ are. 
As illustrated in Figure~\ref{fig:repeation_in_prune}, $C_i(u) = \{a, b, d, e, g, h\}$ and $C_j(u) = \{a, b, c, d, e, f\}$. For simplicity, all neighbors are sorted in the increasing order of their distance to $u$, e.g., $\delta(u, a) \leq \delta(u, b) \leq \dots \leq \delta(u, h)$ in $C_i(u)$. 
After $a$ joins both $C_i(u)$ and $C_j(u)$, $\delta(a, d)$ will be computed by both \texttt{Prune} operations to check is $d$ is dominated by $a$. 
Here, $\delta(a, d)$ is computed repeatedly. 


A naive method to address this issue is to use a hash table, where the vector ID pair $(id_1, id_2)$ of the distance serves as the key, and the distance value is set as the actual distance for its first verification. Then, before each distance computation in each pruning, we first check whether or not it has been in the hash table. However, it is practically inefficient because the size of such a hash table is up to $O(|\cup_{i=1}^{m} C_i(u)|^2)$, where $|C_i(u)|$ could be several hundreds. As a result, the maintenance and queries of such a hash table incur significantly additional cost, including memory allocation for newly inserted elements and the computation of hash keys.

\begin{algorithm}[t]
\small  
\SetFuncSty{textsf}
\SetArgSty{textsf}
\caption{mPrune($u,C_i(u),M, \alpha$)}
\label{alg:mprune}
\Input{a vertex $u$, its candidate neighbor set $C_i(u)$, out-degree limit $M$ and a parameter $\alpha$}
\Output{a pruned neighbor set $C_i^{'}(u)$}
add the first neighbor $u^*$ of $C_i(u)$ to $C_i^{'}(u)$\;
\For{each $v \in C_i(u) \setminus \{u^*\}$}
{
    $Flag\leftarrow$ false\;
    \For{each $w \in C_i^{'}(u)$}
    {
        \If{$v, w \in C_{i-1}^{'}(u)$} {
            \textbf{continue}\;
        }
        \If{$\alpha \cdot \delta(v,w)<\delta(u,v)$}
        {
            $Flag\leftarrow$ true\;
        }
    }
    \If{$Flag=$ false} {
        $C_i^{'}(u)\leftarrow C_i^{'}(u) \cup \{v\}$\;
    }
    \textbf{if} $|C_i^{'}(u)|\ge M$ \textbf{then break}\;
}
\Return{$C_i^{'}(u)$}\;
\end{algorithm}

\vspace{1mm}
\stitle{Fast Pruning for Multiple Graphs.} In this work, we reduce the repeated distance computations between two consecutive pruning operations via only slightly extra storage. We subsequently prune $C_1(u), C_2(u), \ldots, C_m(u)$. 
As in Figure~\ref{fig:repeation_in_prune}, consider that we just obtain $C_i^{'}(u)$ by pruning $C_i(u)$ and then prune $C_{j}(u)$, where $j=i+1$. During the pruning, we compute $\delta(a, d)$,  to check whether or not $d$ is dominated by $a$. However, with $a \in C_i(u) \cap C_{i+1}(u)$ and $d \in C_i^{'}(u)$, we have that $\delta(a, d)$ has been computed before and $a$ does not dominate $d$. In this way, we can avoid computing $\delta(a, d)$ again. Note that the more similar $C_i(u)$ and $C_{i+1}(u)$ are, the more repeated distance computations can be saved. 

We show the details in Algorithm~\ref{alg:mprune}. Note that pruning on $C_1(u)$ is unchanged as in Algorithm~\ref{alg:prune}. However, starting from pruning on $C_2(u)$, we employ our method in Algorithm~\ref{alg:mprune}. Since $C_i(u)$ is sorted in the ascending order of the distance $u$, the first neighbor, i.e., the closest one to $u$, is added into $C_i^{'}(u)$ (line 1). Then, we check whether or not each neighbor $v \in C_i(u)$ is dominated by existing ones in $C_i^{'}(u)$ (lines 2-11). Different from the original pruning method, we avoid computing $\delta(v, w)$ if $v, w \in C_{i-1}^{'}(u)$. This indicates that $v, w$ have been verified in the last pruning operation, and $w$ is not dominated by $v$. Thus, repeated distance computations are saved and avoided.

\begin{algorithm}[t] 
\caption{BuildMultiHNSW($D,P$)}
\label{alg:build_multi_hnsw}
\Input{$D \subset \mathbb{R}^d$ and the candidate parameter set $P$}
\Output{$\mathcal{G} = \{G_1, G_2, \ldots, G_m\}$}
\For{$i \leftarrow 1, \ldots m$}{
    initialize $G_i^{0}$ with a random point $v \in D$\;
    $m_L \leftarrow  0$\;
}
\For{each $u \in D \setminus \{v\}$}{
     randomly determine the highest layer $l$ of $u$\;
     \textbf{if} ($l > m_L$) \textbf{then} $m_L \leftarrow l$ and $ep \leftarrow u$\;
    {initialize $V_{\delta}$ as $-1$ for each vector}\;
    \For{$i \leftarrow 1, \ldots m$} {
         $(efc_i, M_i) \leftarrow P_i$ and $c \leftarrow ep$\;
         \For{each $j \leftarrow m_L$ \text{downto} $l+1$}{
            {$c \leftarrow \text{mKANNS}(G_i^j, u, 1,1, c, V_{\delta})$}\;
         }
         $C_{l+1} \leftarrow \{ c \}$\;
         \For{$j \leftarrow l$ \text{downto} 0}{
            $C_j\gets \text{mKANNS}(G_i^j, u, efc_i, efc_i,C_{j+1}[0], V_{\delta})$\;
            $N_{G_i^j}(u)\gets \text{mPrune}(u, C_j,M_i, 1)$\; 
            \For{each $v \in N_{G_i^j}(u)$}{
                $N_{G_i^j}(v) \leftarrow N_{G_i^j}(v) \cup \{u\}$\;
                \If{$|N_{G_i^j}(v)| > M_i$}{
                $N_{G_i^j}(v) \leftarrow Prune(v, N_{G_i^j}(v), M_i, 1)$\;
                }
            } 
        }
    }
 }
\textbf{return} $\mathcal{G} = \{G_1, \ldots, G_m\}$\;
\end{algorithm}

\begin{algorithm}[t]
\caption{BuildMultiVamana($D,P$)}
\label{alg:build_multi_vamana}
\Input{$D \subset \mathbb{R}^d$ and the candidate parameter set $P$}
\Output{$\mathcal{G} = \{G_1, G_2, \ldots, G_m\}$}
\For{$1 \leq i \leq m$} {
    initialize a random KNNG $G_i$ according to $P_i$\;
}
let $c$ denote the centroid of dataset $D$\;
    \For{each $u \in D$} {
        {initialize $V_{\delta}$ as $-1$ for each point}\;
        \For{$1 \leq i \leq m$} {
         $(L_i,M_i, \alpha_i) \leftarrow P_i$\;
         { $C_i \leftarrow \text{mKANNS}(G_i, u, L_i, L_i, c, V_{\delta})$}\;
       {$N_{G_i}(u) \leftarrow$ mPrune$(u, C_i, M_i, \alpha_i)$ }\;
        \For{each $v$ in $N_{G_i}(u) $} {
            \If{$|N_{G_i}(v)  \cup \{u\}| > M_i$} {
                Prune$(v, N_{G_i}(v) \cup \{ u\},  M_i, \alpha_i)$\;
            }
        }
    }
}
\Return{$\mathcal{G} = \{G_1, \ldots, G_m\}$}\;
\end{algorithm}

\eat{
\begin{algorithm}[t]
\caption{BuildMultiNSG($D,P$)}
\label{alg:build_multi_nsg}
\Input{$D \subset \mathbb{R}^d$ and the candidate parameter set $P$}
\Output{$\mathcal{G} = \{G_1, G_2, \ldots, G_m\}$}
\For{$1 \leq i \leq m$} {
    initialize a KNNG $G_i^0$ according to $P_i$\;
}
let $c$ denote the centroid of dataset $D$\;
    \For{each $u \in D$} {
        {initialize $V_{\delta}$ as $-1$ for each point}\;
        \For{$1 \leq i \leq m$} {
         $(L_i,R_i,M_i) \leftarrow P_i$\;
         {$C_i \leftarrow \text{mKANNS}(G_i^0, u, R_i, L_i, c, V_{\delta})$}\;
       {$N_{G_i}(u) \leftarrow$ mPrune$(u, C_i, M_i, 1)$ }\;
        \For{each $v$ in $N_{G_i}(u) $} {
            \If{$|N_{G_i}(v)  \cup \{u\}| > M_i$} {
                Prune$(v, N_{G_i}(v) \cup \{ u\},  M_i, 1)$\;
            }
        }
    }
}
\Return{$\mathcal{G} = \{G_1, \ldots, G_m\}$}\;
\end{algorithm}
}

\subsection{Efficient Construction of Multiple PGs}
\label{ssec:m_pg}

This part delves into the efficient construction of multiple PGs with respect to a parameter set $P = \{P_1, P_2, \ldots, P_m\}$ through our efficient multiple \texttt{Search} and multiple \texttt{Prune} methods. As previously mentioned, we consider three prominent PGs: HNSW~\cite{Hnsw}, NSG~\cite{Nsg}, and Vamana~\cite{diskann}.

\vspace{1mm}
\stitle{Efficient Construction of Multiple HNSW Graphs.} We present the details in Algorithm~\ref{alg:build_multi_hnsw}.
Notably, $ef$ and $M$ are two key construction parameters of HNSW. First, we initialize each HNSW graph (lines 1-3), where $m_L$ indicates the highest layer of the current HNSW graphs. Then, we insert each $u \in D$ into the graphs iteratively~(lines 4-19). In the loop, we first determine the highest layer $l$ for $u$ (line 5), which is used for all $m$ graphs, and then update $m_L$ if necessary (line 6). We initialize $V_{\delta}$ (line 7) and then insert $u$ into the $m$ graphs one by one (lines 8-19). For the layers from $m_L$ to $l+1$, we find the 1-ANN of $u$ in each layer, and use it as the entry point of the next layer (line 11). Here, $G_i^j$ represents the $j$-th layer of the $i$-th graph $G_i$. For the remaining layers from $l$ to 0, we find their $ efc_j$-ANNs using our efficient mKANNS method (line 14), followed by our mPrune method (line 15). Then, we add reverse edges from the pruned neighbors to $u$ (lines 17) and ensure that the node's out-degree limit is met (lines 18-19). Note that this method could be easily extended to other NSWGs such as NSW~\cite{Sw}, where we fix $m_L$ as 1, employ the closeness-first pruning method, and remove the out-degree limit. Note that $V_{\delta}$ will be deleted after inserting $u$ into all $m$ graphs, since storing $n$ $V_{\delta}$ arrays requires $O(n^2)$ space. 

\vspace{1mm}
\stitle{Efficient Construction of Multiple Vamana Graphs.} We present the details in Algorithm~\ref{alg:build_multi_vamana}. First, we build $m$ random graphs, $G_1^0, G_2^0, \ldots, G_m^0$, with different out-degrees according to each construction parameter (lines 1-2). Let $c$ be the centroid of $D$, which is used as the entry point of $k$-ANNS (line 3). Then, we insert each $u \in D$ into the $m$ graphs (lines 4-12), starting by initializing $V_{\delta}$ (line 5). To insert $u$ into each graph $G_i$ (lines 6-12), we first obtain $C_i(u)$ (line 8) by our mKANNS method (line 8), followed by our mPrune method (line 9). We guarantee the out-degree limit for each node when inserting the edges from neighbors in $N_{G_i}(u)$ to $u$ (lines 10-12). 

\vspace{1mm}
\stitle{Efficient Construction of Multiple NSG Graphs.} 
We can extend Algorithm~\ref{alg:build_multi_vamana} for other RNGs by only modifying the construction of the initial graph and the pruning strategy. Take NSG as an example,
we first build a KNNG $G_i^0$ by KGraph~\cite{KGraph} in line 2 instead of a random one. Next, we employ the KNNG $G_i^0$ to search for line 8, and fix $\alpha_i = 1$ for each parameter, since NSG employs $\alpha_i = 1$ in any case.

\begin{table}[t]
\caption{Statistics of Datasets.}
\label{tb:data}
\begin{center}
\begin{tabular}{|c|r|r|r|c|}
\hline
\textbf{Dataset}& \textbf{Dim.} & \textbf{\#vectors} & \textbf{\#queries} & \textbf{Type}\\
\hline \hline
\text{Sift}& \text{128} & \text{1,000,000}& \text{1,000}& \text{Image}\\
\hline
\text{Gist}& \text{960} & \text{1,000,000}& \text{1,000}& \text{Image}\\
\hline
\text{Glove}& \text{100} & \text{1,183,514}& \text{1,000}& \text{Text}\\
\hline
\text{Msong}& \text{420} & \text{992,272}& \text{200}& \text{Audio}\\
\hline
\end{tabular}
\end{center}
\end{table}

\begin{table*}[ht]
\caption{Comparing FastPGT with Baseline Methods in Parameter Tuning Efficiency.}
\label{tb:Tuning Efficiency}
\begin{center}
\renewcommand{\arraystretch}{1.2}
\begin{tabular}{|c|c|c|rr|rr|rr|}
\hline
\multirow{2}{*}{\bf Datasets} & \multirow{2}{*}{\bf Methods} & \multirow{2}{*}{\bf \#Parameters} & \multicolumn{2}{c|}{\bf HNSW}         & \multicolumn{2}{c|}{\bf NSG}   & \multicolumn{2}{c|}{\bf Vamana}             \\ \cline{4-9} 
  &   &    & \multicolumn{1}{c|}{$\#dist$ ($\times 10^9$)}   & {$cost$ (sec)}    & \multicolumn{1}{c|}{$\#dist$ ($\times 10^9$)}   & {$cost$ (sec)}    & \multicolumn{1}{c|}{$\#dist$ ($\times 10^9$)}   & {$cost$ (sec)}    \\ \hline \hline
\multirow{4}{*}{Gist}  & RandomSearch     & 100     & \multicolumn{1}{r|}{\num{452}} & 40,572 & \multicolumn{1}{r|}{\num{521}} & 34,867   & \multicolumn{1}{r|}{\num{833}} & 45,670 \\ \cline{2-9} 
  & OtterTune   & 100    & \multicolumn{1}{r|}{\num{413}} & 39,891 & \multicolumn{1}{r|}{\num{490}} & 32,234   & \multicolumn{1}{r|}{\num{772}} & 43,129 \\ \cline{2-9} 
  & VDTuner    & 100  & \multicolumn{1}{r|}{\lightcolor\num{382}} & \lightcolor33,006 & \multicolumn{1}{r|}{\lightcolor\num{453}} & \lightcolor30,154 & \multicolumn{1}{r|}{\lightcolor\num{721}} & \lightcolor42,882 \\ \cline{2-9} 
    & FastPGT     & 100   & \multicolumn{1}{r|}{\deepcolor{\num{189}}} & {\deepcolor{15,045}} & \multicolumn{1}{r|}{\deepcolor{\num{143}}} & {\deepcolor{12,697}} & \multicolumn{1}{r|}{\deepcolor{\num{211}}} & {\deepcolor{18,278}} \\ \hline
\multirow{4}{*}{Sift}  & RandomSearch     & 100  & \multicolumn{1}{r|}{\num{567}}         &4,594   & \multicolumn{1}{r|}{\num{401}}     &4,108   & \multicolumn{1}{r|}{\num{263}}   &4,854   \\ \cline{2-9} 
  & OtterTune  & 100  & \multicolumn{1}{r|}{564}  & 4,529  & \multicolumn{1}{r|}{\num{390}}  &4,010 & \multicolumn{1}{r|}{252} &4,773  \\ \cline{2-9} 
   & VDTuner   & 100  & \multicolumn{1}{r|}{\lightcolor407}  & \lightcolor3,282  & \multicolumn{1}{r|}{\lightcolor349}  & \lightcolor3,424   & \multicolumn{1}{r|}{\lightcolor211}  &\lightcolor4,045       \\ \cline{2-9} 
    & FastPGT   & 100   & \multicolumn{1}{r|}{\deepcolor{\num{184}}} & \deepcolor{2,071}  & \multicolumn{1}{r|}{\deepcolor{\num{163}}}      &\deepcolor{2,010}   & \multicolumn{1}{r|}{\deepcolor{\num{133}}}    &\deepcolor{2,336}   \\ \hline
\multirow{4}{*}{Glove} & RandomSearch     & 100   & \multicolumn{1}{r|}{\num{732}}     & 18,987   & \multicolumn{1}{r|}{2,304}   &18,018    & \multicolumn{1}{r|}{1,320}     & 17,437    \\ \cline{2-9} 
& OtterTune   & 100  & \multicolumn{1}{r|}{\num{729}} &18,321  & \multicolumn{1}{r|}{2,297} &17,717  & \multicolumn{1}{r|}{1,301} &16,749   \\ \cline{2-9} 
& VDTuner   & 100  & \multicolumn{1}{r|}{\lightcolor653}  &\lightcolor16,130   & \multicolumn{1}{r|}{\lightcolor1,898}    &\lightcolor15,015    & \multicolumn{1}{r|}{\lightcolor1,104}         &\lightcolor16,366   \\ \cline{2-9} 
& FastPGT   & 100  & \multicolumn{1}{r|}{\deepcolor{\num{383}}}  &\deepcolor{10,711}  & \multicolumn{1}{r|}{\deepcolor{\num{389}}}  & \deepcolor{10,485} & \multicolumn{1}{r|}{\deepcolor{\num{332}}}  & \deepcolor{8,421}  \\ \hline
\multirow{4}{*}{Msong} & RandomSearch & 100 & \multicolumn{1}{r|}{\lightcolor\num{641}}  & \lightcolor12,315  & \multicolumn{1}{r|}{\lightcolor\num{310}}  &\lightcolor9,421  & \multicolumn{1}{r|}{\lightcolor\num{931}}  & \lightcolor13,540    \\ \cline{2-9} 
& OtterTune  & 100 & \multicolumn{1}{r|}{\num{727}}    &10,898  & \multicolumn{1}{r|}{\num{432}} &13,717    & \multicolumn{1}{r|}{1,201}  &  15,246\\ \cline{2-9} 
& VDTuner  & 100  & \multicolumn{1}{r|}{\num{654}} & 12,390 & \multicolumn{1}{r|}{\num{333}}  &9,698   & \multicolumn{1}{r|}{\num{976}} &14,349    \\ \cline{2-9} 
& FastPGT  & 100 & \multicolumn{1}{r|}{\deepcolor{\num{150}}}  &\deepcolor{6,164}  & \multicolumn{1}{r|}{\deepcolor{\num{67}}}  & \deepcolor{4,615} & \multicolumn{1}{r|}{\deepcolor{\num{299}}}   & \deepcolor{6,673}  \\ \hline
\end{tabular}
\end{center}
\end{table*}

\section{Experiments}
\label{sec:exp}

In this section, we present our experimental results. We show the experimental settings and results in order to demonstrate the superiority of our method over its competitors. 

\subsection{Experimental Settings}
\label{ssec:exps}

\stitle{Datasets.} 
We conduct 4 public datasets, Sift, Gist, Glove and Msong, which are widely used to evaluate the performance of $k$-ANNS methods. Their data statistics are shown in Table \ref{tb:data}. Here, we use dim. to represent the dimensionality of the vectors and \#vectors as the number of vectors in $D$. Let $Q$ be the set of queries and \#queries the number of queries in $Q$.
To be specific, \textbf{Sift} \footnote{\label{fn:bigann}\textit{http://corpus-texmex.irisa.fr}.} contains 1,000,000 128-dimensional SIFT vectors.  \textbf{Gist} \footref{fn:bigann} consists of 1,000,000 960-dimensional GIST vectors. \textbf{Glove} \footnote{\textit{http://nlp.stanford.edu/projects/glove/}.} comprises 1,183,514 100-dimensional word feature vectors extracted from Tweets. Msong~\footnote{\textit{http://www.ifs.tuwien.ac.at/mir/msd/download.html}.} contains 994,185 420-dimensional Temporal Rhythm Histograms extracted from the same number of contemporary popular music tracks. For each dataset, we randomly select a specified number of points, as presented in Table \ref{tb:data}, to form the query set $Q$, which is used to evaluate the $k$-ANNS performance of PGs generated according to the construction parameters.


\stitle{Performance Indicators.} 
We are concerned with two performance aspects of each tuning method, i.e., tuning efficiency and tuning quality. The former is estimated using two performance indicators: the time cost incurred and the number of distance computations during the tuning process. The tuning quality can be evaluated by the $k$-ANNS performance of the PG generated according to the finally returned parameter. As aforementioned, the $k$-ANNS performance could be estimated by two performance indicators, i.e., Queries Per Second (QPS) and $Recall@k$. 
The reported $Recall@k$ is averaged over the entire query set $Q$. By default, we set $k=10$ unless specified. All reported results are averaged over 3 independent runs.


\stitle{Environments.} We implement all methods in \texttt{C++} and compile them with \texttt{g++ 11}. We conduct all the experiments on a server equipped with two Intel(R) Xeon(R) Gold 6240 CPUs, each of which has 18 cores and 36 hyper-threads, and 380 GB of memory. Its operating system is \texttt{CentOS 7.6}. 

\stitle{Baselines.} We compare our method FastPGT with several SOTA a parameter tuning methods, i.e., RandomSearch~\cite{Bergstra_Bengio_2012}, VDTuner~\cite{yang2024vdtuner} and OtterTune~\cite{van2017automatic}. To make a fair comparison, we set the tuning budget for each method to explore 100 parameter candidates. We set the batch size to 10 for FastPGT, which means that it recommends 10 parameter candidates in each iteration. Here, we omit GridSearch due to its prohibitively high computational cost, which grows exponentially with the number of construction parameters, making it impractical for our tuning task.  




\subsection{Main Results}
\label{ssec:mainr}

\begin{figure*}[t]
\centering
\includegraphics[width=0.99\linewidth]{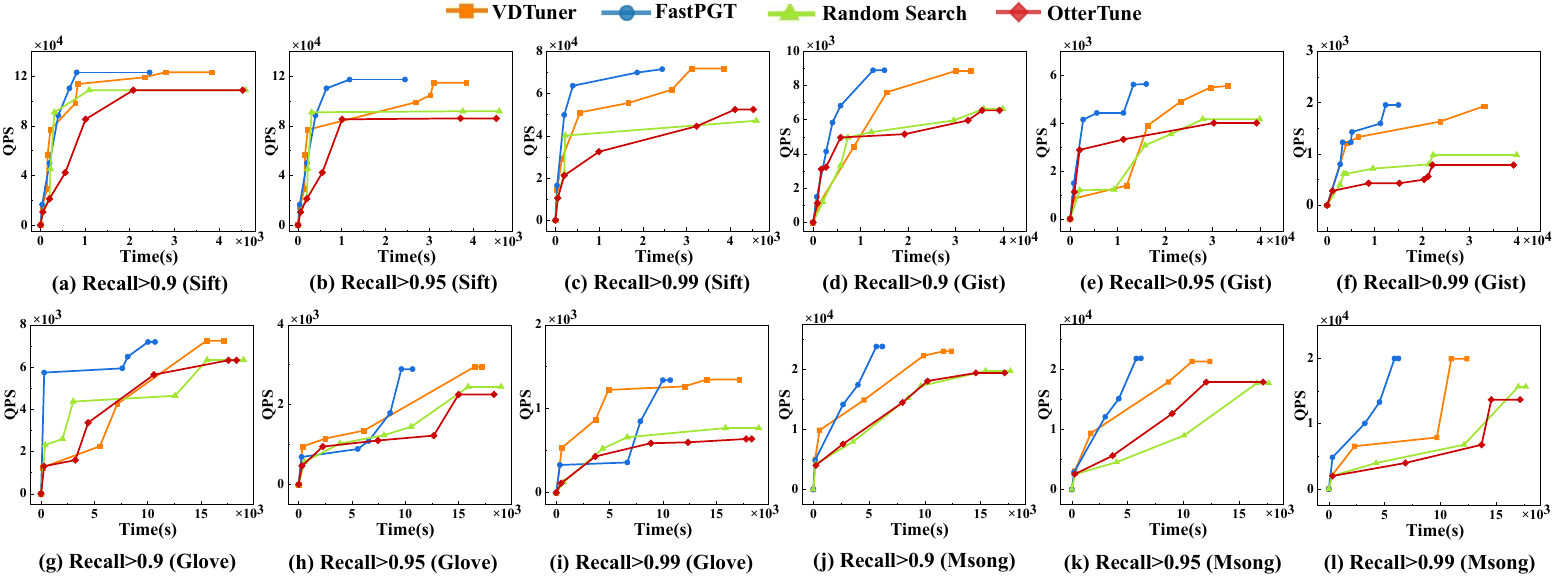}
\caption{Comparison of the performance of our method and the baseline in tuning HNSW.}
\label{fig:exp_hnsw_tuning}
\end{figure*}

\begin{figure*}[t]
\centering
\includegraphics[width=0.99\linewidth]{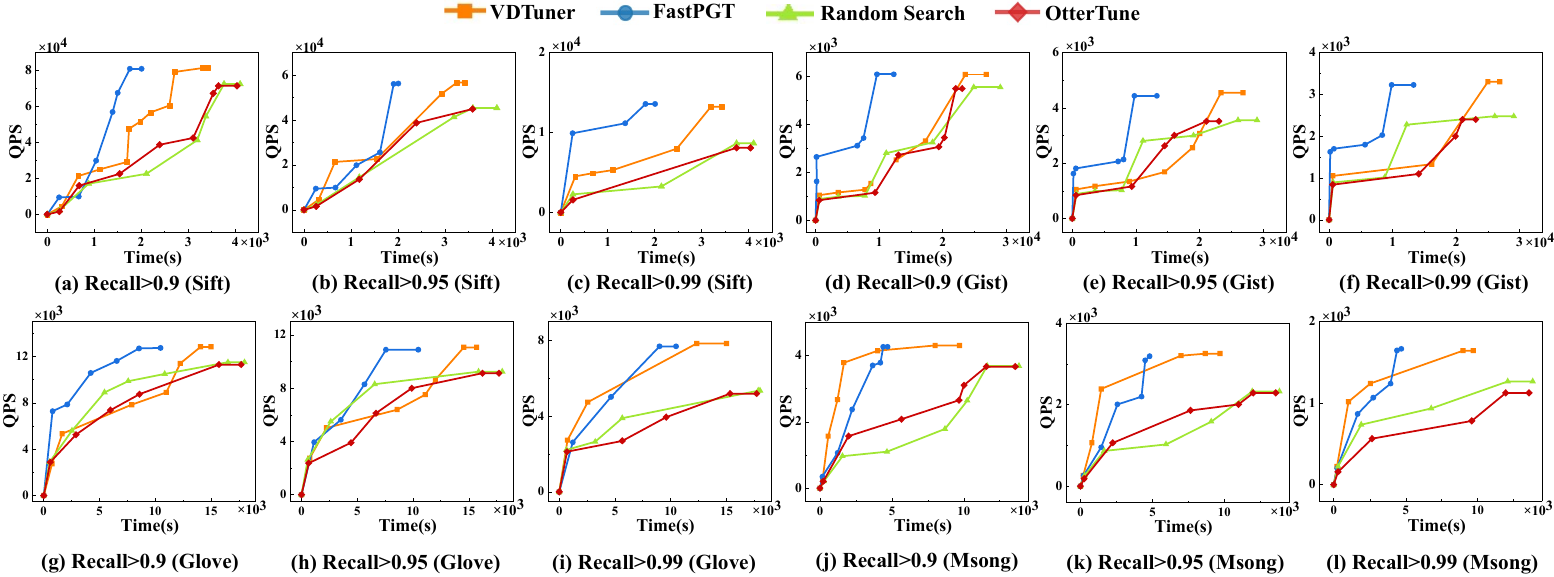}
\caption{Comparison of the performance of our method and the baseline in tuning NSG.}
\label{fig:exp_nsg_tuning}
\end{figure*}

\begin{figure*}[t]
\centering
\includegraphics[width=0.99\linewidth]{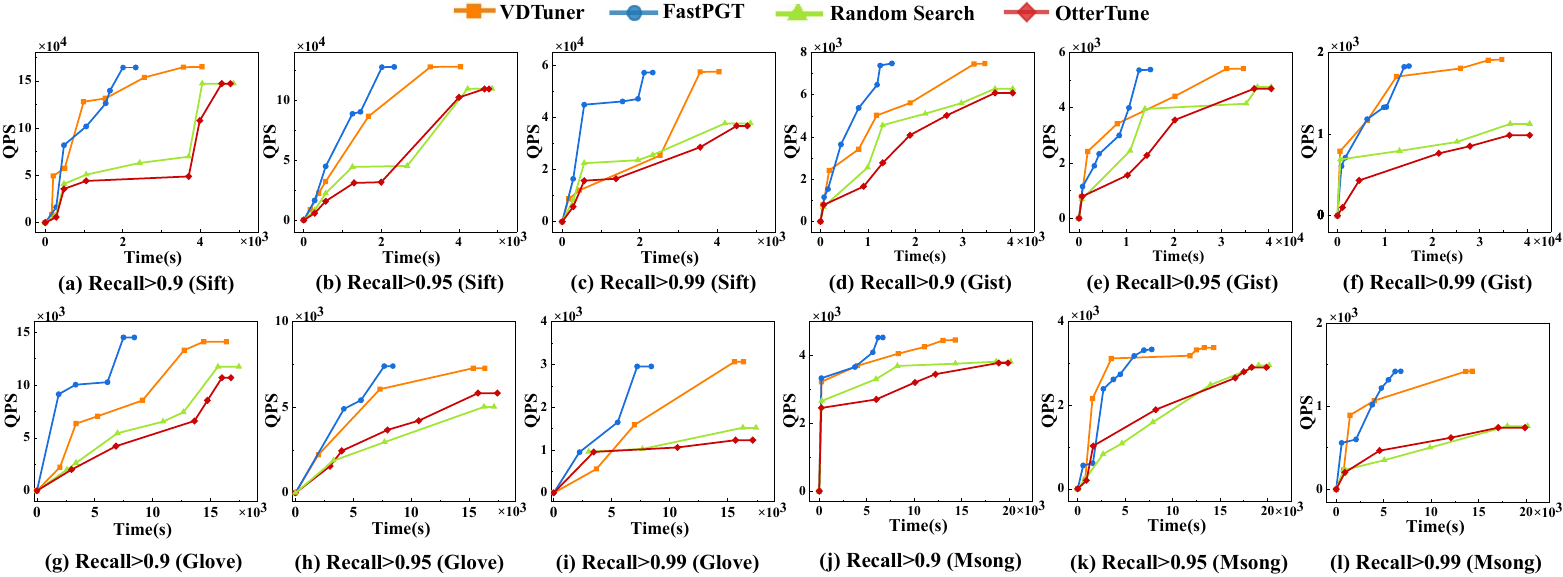}
\caption{Comparison of the performance of our method and the baseline in tuning Vamana.}
\label{fig:exp_vamana_tuning}
\end{figure*}

In this part, we show the main experimental results, where we tune the construction parameters for three representative PGs, i.e., HNSW, NSG, and Vamana. 

\stitle{Tuning Efficiency.} 
Given the same number of parameters recommended and estimated, we compare the tuning efficiency between our method FastPGT and its competitors in tuning efficiency in Table \ref{tb:Tuning Efficiency}. 
We can see that FastPGT obviously achieves less tuning cost than its competitors. In particular, FastPGT obtains {2.2x, 2.37x, 2.35x} speedup over the SOTA method VDTuner in time cost, when tuning HNSW, NSG, and Vamana on Gist, respectively. Moreover, similar phenomena could be found on other datasets.  
This is because our method effectively reduces the repeated distance computations in parameter estimation. 
To be specific, FastPGT only computes 50\%, 31.1\%, and 29.2\% of distances of HNSW, NSG, and Vamana on Gist, respectively, compared with VDTuner. Moreover, similar phenomena can also be observed in other datasets and baselines. 
As a result, our method delivers clear speedups during the tuning process. 

\stitle{Tuning Quality.} 
Given the time budget for tuning, we compare the performance of $k$-ANNS between FastPGT and its competitors. We show the experimental results on HNSW, NSG, and Vamana in Figure~\ref{fig:exp_hnsw_tuning}, \ref{fig:exp_nsg_tuning} and \ref{fig:exp_vamana_tuning}, respectively. Here, we consider three $Recall@k$ values — 0.9, 0.95, and 0.99 —and use the corresponding QPS values to measure the $k$-ANNS performance. 
With the same budget, FastPGT presents better tuning quality than its competitors in almost all settings. 
In particular, take NSG on Gist as an example in Figure \ref{fig:exp_nsg_tuning}. FastPGT spends only 37.8\% of the tuning cost of VDtuner to obtain comparable tuning quality across various $Recall@k$ values, and consumes only 9.81\%, 19.82\%, and 23.22\% of the tuning cost of OtterTuner to obtain comparable or even better tuning quality with the three $Recall@k$ values, respectively. 
Hence, it demonstrates that FastPGT incurs much lower tuning costs to achieve comparable or even better tuning quality than its competitors. 


\subsection{Ablation Study}
\label{ssec:ablation}
The tuning efficiency of FastPGT is mainly attributed to our efficient multiple \texttt{Search} operations (ESO) and efficient multiple \texttt{Prune} method (EPO).
In this part, we study their contributions to the overall cost. 
We design three configurations of FastPGT, as summarized in Table \ref{tb:ablation}, Config (I) disables both ESO and EPO, Config (II) activates ESO only, and Config (III) enables both ESO and EPO. 
Note that all of them have the same tuning quality but different tuning efficiency. Hence, we compare them in tuning efficiency with two performance metrics: (1) relative tuning cost (RTC) to Config (I) and (2) relative number of distance computations (RDC) to Config (I). 
Both RTC and RDC of Config (I) are 1. 
We present the results on Msong in Table~\ref{tb:ablation}. 

\stitle{Effectiveness of the ESO.}
By comparing Config (I) and (II), the introduction of ESO yields substantial performance improvements. Specifically, Config (II) with ESO only computes 39\%, 44\% and 57\% distances of Config (I) on NSG, HNSW, and Vamana, respectively, and thus requires 54\%, 52\% and 54\% tuning cost of Config (I,) respectively. 

\stitle{Effectiveness of the EPO.}
Config (III) with EPO obviously computes fewer distances and thus less tuning cost than Config (II) without EPO, since EPO reduces the repeated distance computations in \texttt{Prune} operations. 
To be specific, let us consider NSG. Config (III) 67 billion distances, which is only 52\% of Config (II). Further, the tuning cost of Config (III) is 0.87 of Config (II).

\begin{table}[t]
\setlength{\tabcolsep}{3pt}
\renewcommand{\tabcolsep}{1.7pt}
\caption{The effects of ESO and EPO. }
\label{tb:ablation}
\begin{center}
\begin{tabular}{|c|ccc|r|r|r|r|}
\hline
\textbf{Index} & \textbf{Config} & \textbf{ESO} &  \textbf{EPO}  & \textbf{$cost$ (sec)} & \textbf{$\#dist$ $(\times 10^9)$} & \textbf{RTC} & \textbf{RDC} \\
\hline \hline
  &(I) & $\times$  & $\times$ &9,698  & \num{331} & 1 & 1 \\
NSG  &(II) & \checkmark  & $\times$ &5,209 & \num{129} & 0.54 & 0.39 \\
  &(III) & \checkmark  & \checkmark &4,615  &\num{67}& 0.47 & 0.18 \\
\hline
  &(I) & $\times$  & $\times$ &12,390  &\num{531} &  1 & 1 \\
HNSW  &(II) & \checkmark  & $\times$ &6,532  &	\num{229} & 0.52 & 0.44 \\
  &(III) &\checkmark  & \checkmark &6,164  &\num{150} & 0.49 & 0.29 \\
\hline
 &(I) & $\times$  & $\times$ &14,349  &\num{982}& 1 & 1 \\
Vamana &(II) & \checkmark  & $\times$ &7,797  &\num{558} & 0.54 & 0.57 \\
  &(III) & \checkmark  & \checkmark &6,673  &\num{299} & 0.47 & 0.31 \\
\hline
\end{tabular}
\end{center}
\end{table}

\subsection{Extensions to Other Recommendation Models}
\label{sec:effectiveness_other_models}

Notably, our method for accelerating the construction of multiple PGs is model-agnostic, independent of the recommendation model, and can be seamlessly integrated into models that enable batch recommendations. 
In this part, we combine ESO and EPO with RandomSearch, denoted as RandomSearch$^+$ (RS$^+$). RandomSearch first generates its batch of candidate parameters, and our methods are then employed to build the corresponding PGs simultaneously. 

\begin{table}[t]
\caption{Extending our method to RandomSearch (RS). RS$^+$ denotes the RS enhanced with both ESO and EPO.}
\label{tb:RS}
\begin{center}
 \resizebox{0.49\textwidth}{!}{
\begin{tabular}{|c|c|r|r|r|r|}
\hline
\textbf{Datasets}&\textbf{Methods} & \textbf{$cost$ (sec)} & \textbf{\makecell{$\#dist$ $(\times 10^9)$ }}& \textbf{RTC}& \textbf{RDC} \\
\hline \hline 
\multirow{2}{*}{Sift} 
&\text{RS}& 4,594 & \num{567} &1 &1 \\ 
& RS$^+$ &1,975 & \num{97} &0.43 &0.17 \\ 
\hline
\multirow{2}{*}{Gist}
&\text{RS}& 40,572 &\num{452} & 1 &1\\ 
& RS$^+$ &13,794 &\num{68} &0.34 &0.15\\ 
\hline
\multirow{2}{*}{Glove}
&\text{RS}& 18,987 &  \num{732} &1 &1 \\ 
& RS$^+$ &9,873  &\num{154}  &0.52 &0.21 \\\cline{2-6}
\hline
\multirow{2}{*}{Msong}
&\text{RS}& 12,315 & \num{641} &1 &1\\ 
& RS$^+$ &5,049 &\num{122}  &0.41 &0.19\\ 
\hline
\end{tabular}
}
\end{center}
\end{table}

As shown in Table~\ref{tb:RS}, RandomSearch$^+$ obviously speeds up RandomSearch. To be specific, RandomSearch$^+$ only consumes 34\%-52\% of the tuning cost of RandomSearch, since it reduces the distance computations. RandomSearch$^+$ only computes 15\%-21\% of the distance computations of RandomSearch.

\section{Conclusion}
\label{sec:con}
In this paper, we study the tuning of the construction parameters of proximity graphs (PG), which are the SOTA methods for $k$-ANNS. We aim to efficiently recommend high-quality construction parameters for a given dataset and a specific type of PG. 
To address this problem, we propose a novel tuning framework, FastPGT, comprising two steps: parameter recommendation and parameter estimation. We design a new parameter recommendation model that recommends a batch of parameters per iteration based on the SOTA method, VDTuner. 
Further, we build multiple PGs simultaneously in each parameter estimation, which is equipped with our efficient multiple \texttt{Search} and \texttt{Prune} operations.
Notably, our efficient parameter estimation method is both model-agnostic and PG-agnostic. 
We conduct extensive experiments on real datasets to demonstrate the superiority of our method FastPGT over its competitors. 
According to the results, FastPGT consumes much lower tuning cost than its competitors while achieving comparable or even better tuning quality on three representative PGs, i.e., HNSW, Vamana, and NSG.

\bibliographystyle{IEEEtran}
\bibliography{./IEEEabrv,main}

@article{liu2024retrievalattentionacceleratinglongcontextllm,
  title={Retrievalattention: Accelerating long-context llm inference via vector retrieval},
  author={Liu, Di and Chen, Meng and Lu, Baotong and Jiang, Huiqiang and Han, Zhenhua and Zhang, Qianxi and Chen, Qi and Zhang, Chengruidong and Ding, Bailu and Zhang, Kai and others},
  journal={arXiv preprint arXiv:2409.10516},
  year={2024}
}

@inproceedings{deng2025alayadb,
  title={{AlayaDB}: The Data Foundation for Efficient and Effective Long-context LLM Inference},
  author={Deng, Yangshen and You, Zhengxin and Xiang, Long and Li, Qilong and Yuan, Peiqi and Hong, Zhaoyang and Zheng, Yitao and Li, Wanting and Li, Runzhong and Liu, Haotian and others},
  booktitle={Companion of SIGMOD},
  pages={364--377},
  year={2025}
}

@article{zhang2025pqcache,
  title={Pqcache: Product quantization-based kvcache for long context llm inference},
  author={Zhang, Hailin and Ji, Xiaodong and Chen, Yilin and Fu, Fangcheng and Miao, Xupeng and Nie, Xiaonan and Chen, Weipeng and Cui, Bin},
  journal={Proceedings of the ACM on Management of Data},
  volume={3},
  number={3},
  pages={1--30},
  year={2025},
  publisher={ACM New York, NY, USA}
}

@inproceedings{KGraph,
	author = {Dong, Wei and Moses, Charikar and Li, Kai},
	booktitle = {WWW},
	pages = {577--586},
	title = {Efficient k-nearest neighbor graph construction for generic similarity measures},
	year = {2011}}

@inproceedings{MyKNNGCsurvey,
	author = {Liu, Yingfan and Cheng, Hong and Cui, Jiangtao},
	booktitle = {arXiv preprint arXiv:2112.02234},
	title = {Revisiting k-Nearest Neighbor Graph Construction on High-Dimensional Data : Experiments and Analyses},
	year = {2021}}

@article{Sw,
  author  = {Yury Malkov and Alexander Ponomarenko and Andrey Logvinov and Vladimir Krylov},
  title   = {Approximate nearest neighbor algorithm based on navigable small world graphs},
  journal = {Information Systems},
  year    = {2014},
  volume  = {45},
  pages   = {61-68},
}

@article{Hnsw,
	author = {Malkov, Yury and Yashunin, Dmitry},
	journal = {IEEE TPAMI},
	publisher = {IEEE},
	title = {Efficient and robust approximate nearest neighbor search using hierarchical navigable small world graphs},
	pages = { 824--836},
	volume = {42},
	 number = {4},
	year = {2018}}

@article{Dpg,
	author = {Li, Wen and Zhang, Ying and Sun, Yifang and Wang, Wei and Li, Mingjie and Zhang, Wenjie and Lin, Xuemin},
	journal = {IEEE TKDE},
	publisher = {IEEE},
	title = {Approximate nearest neighbor search on high dimensional data -- experiments, analyses, and improvement},
	volume = {32},
	number = {8},
	pages = {1475--1488},
	year = {2019}}

@article{Nsg,
	author = {Fu, Cong and Xiang, Chao and Wang, Changxu and Cai, Deng},
	journal = {PVLDB},
	number = {5},
	pages = {461--474},
	publisher = {VLDB Endowment},
	title = {Fast approximate nearest neighbor search with the navigating spreading-out graph},
	volume = {12},
	year = {2019}}

@article{Nssg,
author = {Fu, Cong and Wang, Changxu and Cai, Deng},
title = {High Dimensional Similarity Search With Satellite System Graph: Efficiency, Scalability, and Unindexed Query Compatibility},
year = {2022},
issue_date = {Aug. 2022},
publisher = {IEEE Computer Society},
address = {USA},
volume = {44},
number = {8},
journal = {IEEE TPAMI},
pages = {4139–4150},
numpages = {12}
}

@inproceedings{HuangSSXZPPOY20,
  author       = {Jui{-}Ting Huang and
                  Ashish Sharma and
                  Shuying Sun and
                  Li Xia and
                  David Zhang and
                  Philip Pronin and
                  Janani Padmanabhan and
                  Giuseppe Ottaviano and
                  Linjun Yang},
  title        = {Embedding-based Retrieval in Facebook Search},
  booktitle    = {{KDD}},
  pages        = {2553--2561}
}

@inproceedings{li2021embedding,
  title={Embedding-based product retrieval in taobao search},
  author={Li, Sen and Lv, Fuyu and Jin, Taiwei and Lin, Guli and Yang, Keping and Zeng, Xiaoyi and Wu, Xiao-Ming and Ma, Qianli},
  booktitle={SIGKDD},
  pages={3181--3189},
  year={2021}
}

@inproceedings{OkuraTOT17,
  author       = {Shumpei Okura and
                  Yukihiro Tagami and
                  Shingo Ono and
                  Akira Tajima},
  title        = {Embedding-based News Recommendation for Millions of Users},
  booktitle    = {{SIGKDD}},
  pages        = {1933--1942},
  publisher    = {{ACM}},
  year         = {2017}
}

@article{survey2021,
author = {Wang, Mengzhao and Xu, Xiaoliang and Yue, Qiang and Wang, Yuxiang},
title = {A Comprehensive Survey and Experimental Comparison of Graph-Based Approximate Nearest Neighbor Search},
year = {2021},
journal = {PVLDB},
issue_date = {July 2021},
publisher = {VLDB Endowment},
volume = {14},
number = {11},
pages = {1964–1978},
numpages = {15}
}

@article{Kdtree,
  title={Multidimensional binary search trees used for associative searching},
  author={Bentley, Jon Louis},
  journal={Communications of the ACM},
  volume={18},
  number={9},
  pages={509--517},
  year={1975},
}

@INPROCEEDINGS{Rtree,
  author = {Antonin Guttman},
  title = {R-trees: a dynamic index structure for spatial searching},
  booktitle = {SIGMOD},
  year = {1984},
  pages = {47--57}, 
}

@INPROCEEDINGS{Xtree,
  author = {Stefan Berchtold and Daniel A. Keim and Hans-Peter Kriegel},
  title = {The {X}-tree: An Index Structure for High-Dimensional Data},
  booktitle = {VLDB},
  year = {1996},
  pages = {28--39}
}

@INPROCEEDINGS{Srtree,
  author = {Katayama, Norio and Satoh, Shin'ichi},
  title = {The SR-tree: an index structure for high-dimensional nearest neighbor
  queries},
  booktitle = {SIGMOD},
  year = {1997},
  pages = {369 - 380}
}

@ARTICLE{PQ,
  author = {Herve Jegou and Matthijs Douze and Cordelia Schmid},
  title = {Product quantization for nearest neighbor search},
  journal = {IEEE TPAMI},
  year = {2011},
  volume = {33(1)},
  pages = {117-128},
}

@article{IMI,
  title={The inverted multi-index},
  author={Babenko, Artem and Lempitsky, Victor},
  journal={IEEE TPAMI},
  volume={37},
  number={6},
  pages={1247--1260},
  year={2014},
  publisher={IEEE}
}

@inproceedings{diskann,
author = {Subramanya, Suhas Jayaram and Devvrit  and Kadekodi, Rohan and Krishnaswamy, Ravishankar and Simhadri, Harsha},
title = {DiskANN: Fast Accurate Billion-point Nearest Neighbor Search on a Single Node},
booktitle = {NeurIPS},
year = {2019}
}

@INPROCEEDINGS{Datar2004,
  author = {Datar, Mayur and Immorlica, Nicole and Indyk, Piotr and Mirrokni,
	Vahab S.},
  title = {Locality-sensitive hashing scheme based on p-stable distributions},
  booktitle = {SoCG},
  year = {2004},
  pages = {253-262}
}

@INPROCEEDINGS{MPlsh,
  author = {Lv, Qin and Josephson, William and Wang, Zhe and Charikar, Moses
  and Li, Kai},
  title = {{Multi-probe LSH}: efficient indexing for high-dimensional similarity
  search},
  booktitle = {VLDB},
  year = {2007},
  pages = {950--961}
}

@INPROCEEDINGS{LSB,
  author = {Tao, Yufei and Yi, Ke and Sheng, Cheng and Kalnis, Panos},
  title = {Quality and efficiency in high dimensional nearest neighbor search},
  booktitle = {SIGMOD},
  year = {2009},
  pages = {563--576}, 
}

@INPROCEEDINGS{C2LSH,
  author = {Gan, Junhao and Feng, Jianlin and Fang, Qiong and Ng, Wilfred},
  title = {Locality Sensitive Hashing Scheme Based on Dyanmic Collision Counting},
  booktitle = {SIGMOD},
  year = {2012},
  pages = {541-552}
}

@ARTICLE{SKLSH,
  author = {Yingfan Liu and Jiangtao Cui and Zi Huang and Hui Li and Hengtao
  shen},
  title = {{SK-LSH} : an efficient index structure for approximate nearest neighbor
  search},
  journal = {PVLDB},
  year = {2014},
  volume = {7},
  pages = {745-756},
  number = {9},
}

@article{taumg,
  author       = {Yun Peng and
                  Byron Choi and
                  Tsz Nam Chan and
                  Jianye Yang and
                  Jianliang Xu},
  title        = {Efficient Approximate Nearest Neighbor Search in Multi-dimensional
                  Databases},
  journal      = {Proceedings of the ACM on Management of Data},
  volume       = {1},
  number       = {1},
  pages        = {54:1--54:27},
  year         = {2023}
}

@article{ALMG,
author = {Xie, Jiadong and Yu, Jeffrey Xu and Liu, Yingfan},
title = {Graph Based K-Nearest Neighbor Search Revisited},
year = {2025},
publisher = {Association for Computing Machinery},
journal = {ACM TODS},
month = may
}

@inproceedings{liu2025privacy,
  title={Privacy-preserving approximate nearest neighbor search on high-dimensional data},
  author={Liu, Yingfan and Zhang, Yandi and Xie, Jiadong and Li, Hui and Yu, Jeffrey Xu and Cui, Jiangtao},
  booktitle={ICDE},
  pages={3017--3029},
  year={2025},
  organization={IEEE}
}

@article{AumullerBF20,
  author       = {Martin Aum{\"{u}}ller and
                  Erik Bernhardsson and
                  Alexander John Faithfull},
  title        = {ANN-Benchmarks: {A} benchmarking tool for approximate nearest neighbor
                  algorithms},
  journal      = {Information Systems},
  volume       = {87},
  year         = {2020}
}

@inproceedings{LiZAH20,
  author       = {Conglong Li and
                  Minjia Zhang and
                  David G. Andersen and
                  Yuxiong He},
  title        = {Improving Approximate Nearest Neighbor Search through Learned Adaptive
                  Early Termination},
  booktitle    = {SIGMOD},
  pages        = {2539--2554},
  publisher    = {{ACM}},
  year         = {2020}
}

@inproceedings{Cagra,
  title={Cagra: Highly parallel graph construction and approximate nearest neighbor search for gpus},
  author={Ootomo, Hiroyuki and Naruse, Akira and Nolet, Corey and Wang, Ray and Feher, Tamas and Wang, Yong},
  booktitle={ICDE},
  pages={4236--4247},
  year={2024},
  organization={IEEE}
}

@article{RAG-ACL,
  title={Retrieval-based Language Models and Applications},
  author={Asai, Akari and Min, Sewon and Zhong, Zexuan and Chen, Danqi},
  journal={ACL Tutorial},
  year={2023}
}

@inproceedings{REALM,
  title={Retrieval augmented language model pre-training},
  author={Guu, Kelvin and Lee, Kenton and Tung, Zora and Pasupat, Panupong and Chang, Mingwei},
  booktitle={ICML},
  pages={3929--3938},
  year={2020}
}

@misc{sift,
  title        = {Datasets for approximate nearest neighbor search},
  howpublished = {\url{http://corpus-texmex.irisa.fr/}},
  year         = 2010
}

@article{nasrabadi1988image,
  title={Image coding using vector quantization: A review},
  author={Nasrabadi, Nasser M and King, Robert A},
  journal={IEEE Transactions on communications},
  volume={36},
  number={8},
  pages={957--971},
  year={1988},
  publisher={IEEE}
}

@inproceedings{msong,
  author       = {Thierry Bertin{-}Mahieux and
                  Daniel P. W. Ellis and
                  Brian Whitman and
                  Paul Lamere},
  title        = {The Million Song Dataset},
  booktitle    = {Proceedings of the 12th International Society for Music Information
                  Retrieval Conference, {ISMIR} 2011},
  pages        = {591--596},
  publisher    = {University of Miami},
  year         = {2011}
}

@inproceedings{glove,
  author       = {Jeffrey Pennington and
                  Richard Socher and
                  Christopher D. Manning},
  title        = {Glove: Global Vectors for Word Representation},
  booktitle    = {Proceedings of the 2014 Conference on Empirical Methods in Natural
                  Language Processing, {EMNLP} 2014, {A} meeting of SIGDAT, a Special Interest Group of the {ACL}},
  pages        = {1532--1543},
  publisher    = {{ACL}},
  year         = {2014}
}

@article{FastPG,
  author       = {Shuo Yang and
                  Jiadong Xie and
                  Yingfan Liu and
                  Jeffrey Xu Yu and
                  Xiyue Gao and
                  Qianru Wang and
                  Yanguo Peng and
                  Jiangtao Cui},
  title        = {Revisiting the Index Construction of Proximity Graph-Based Approximate Nearest Neighbor Search},
  journal      = {PVLDB},
  volume       = {18},
  number       = {6},
  pages        = {1825--1838},
  year         = {2025}
}

@article{PGTuner,
  author       = {Hao Duan and
                  Yitong Song and
                  Bin Yao and
                  Anqi Liang},
  title        = {PGTuner: An Efficient Framework for Automatic and Transferable Configuration
                  Tuning of Proximity Graphs},
  journal      = {CoRR},
  volume       = {abs/2508.17886},
  year         = {2025}
}

@article{word2vec,
  title={Distributed representations of words and phrases and their compositionality},
  author={Mikolov, Tomas and Sutskever, Ilya and Chen, Kai and Corrado, Greg S and Dean, Jeff},
  journal={NeurIPS},
  volume={26},
  year={2013}
}

@misc{hnswlib,
  title        = {Header-only C++/python library for fast approximate nearest neighbors},
  year         = {2018},
  howpublished = {\url{https://github.com/nmslib/hnswlib}},
  note         = {Accessed: 2025-05-01}
}

@inproceedings{yang2024vdtuner,
  title={Vdtuner: Automated performance tuning for vector data management systems},
  author={Yang, Tiannuo and Hu, Wen and Peng, Wangqi and Li, Yusen and Li, Jianguo and Wang, Gang and Liu, Xiaoguang},
  booktitle={ICDE},
  pages={4357--4369},
  year={2024},
  organization={IEEE}
}

@inproceedings{van2017automatic,
  title={Automatic database management system tuning through large-scale machine learning},
  author={Van Aken, Dana and Pavlo, Andrew and Gordon, Geoffrey J and Zhang, Bohan},
  booktitle={SIGMOD},
  pages={1009--1024},
  year={2017}
}

@article{Bergstra_Bengio_2012,  
 title={Random search for hyper-parameter optimization}, 
 journal={Journal of Machine Learning Research}, 
 author={Bergstra, James and Bengio, Yoshua}, 
 year={2012}, 
 month={Mar}, 
 language={en-US} 
 }

@article{liashchynskyi2019grid,
  title={Grid search, random search, genetic algorithm: a big comparison for NAS},
  author={Liashchynskyi, Petro and Liashchynskyi, Pavlo},
  journal={arXiv preprint arXiv:1912.06059},
  year={2019}
}

@article{Yang_Emmerich_Deutz_Bäck_2019,  
 title={Multi-Objective Bayesian Global Optimization using expected hypervolume improvement gradient}, 
 url={http://dx.doi.org/10.1016/j.swevo.2018.10.007}, 
 DOI={10.1016/j.swevo.2018.10.007}, 
 journal={Swarm and Evolutionary Computation}, 
 author={Yang, Kaifeng and Emmerich, Michael and Deutz, André and Bäck, Thomas}, 
 year={2019}, 
 month={Feb}, 
 pages={945–956}, 
 language={en-US} 
 }

@article{PGsurvey25,
  author       = {Ilias Azizi and
                  Karima Echihabi and
                  Themis Palpanas},
  title        = {Graph-Based Vector Search: An Experimental Evaluation of the State-of-the-Art},
  journal      = {Proceedings of the ACM on Management of Data},
  volume       = {3},
  number       = {1},
  pages        = {43:1--43:31},
  year         = {2025}
}

\end{document}